\documentclass[pdflatex,sn-mathphys-num]{sn-jnl}

\usepackage{graphicx}%
\usepackage{multirow}%
\usepackage{amsmath,amssymb,amsfonts}%
\usepackage{amsthm}%
\usepackage{mathrsfs}%
\usepackage[title]{appendix}%
\usepackage{xcolor}%
\usepackage{textcomp}%
\usepackage{manyfoot}%
\usepackage{booktabs}%
\usepackage{algorithm}%
\usepackage{algorithmicx}%
\usepackage{algpseudocode}%
\usepackage{listings}%
\usepackage{bbm}%
\usepackage{enumitem}%
\usepackage{caption, subcaption}%

\theoremstyle{thmstyleone}%
\newtheorem{theorem}{Theorem}
\newtheorem{proposition}[theorem]{Proposition}%
\newtheorem{assumption}{Assumption}
\newtheorem{lemma}{Lemma}

\theoremstyle{thmstyletwo}%

\theoremstyle{thmstylethree}%

\newcommand{\bR}{\mathbb R}

\begin{document}

\title[Inferring single-cell trajectories from RNA velocity fields with varifold distances]{VeloTree: Inferring single-cell trajectories from RNA velocity fields with varifold distances}


\author*[1, 2]{\fnm{Elodie} \sur{Maignant}}\email{elomai@dtu.dk}

\author[1]{\fnm{Tim} \sur{Conrad}}\email{conrad@zib.de}

\author[1]{\fnm{Christoph} \spfx{von} \sur{Tycowicz}}\email{vontycowicz@zib.de}

\affil[1]{\orgname{Zuse Institute Berlin}, \orgaddress{\street{Takustraße 7}, \city{Berlin}, \postcode{14195}, \country{Germany}}}
\affil[2]{\orgdiv{Department of Applied Mathematics and Computer Science}, \orgname{Technical University of Denmark}, \orgaddress{\street{Richard Petersens Plads, Building 324}, \city{Kgs. Lyngby}, \postcode{2800}, \country{Denmark}}}


\abstract{Trajectory inference is a critical problem in single-cell transcriptomics, which aims to reconstruct the dynamic process underlying a population of cells from sequencing data. Of particular interest is the reconstruction of differentiation trees. One way of doing this is by estimating the path distance between nodes -- labeled by cells -- based on cell similarities observed in the sequencing data. Recent sequencing techniques make it possible to measure two types of data: gene expression levels, and RNA velocity, a vector that quantifies variation in gene expression. The sequencing data then consist in a discrete vector field in dimension the number of genes of interest. In this article, we present a novel method for inferring differentiation trees from RNA velocity fields using a distance-based approach. In particular, we introduce a cell dissimilarity measure defined as the squared varifold distance between the integral curves of the RNA velocity field, which we show is a robust estimate of the path distance on the target differentiation tree. Upstream of the dissimilarity measure calculation, we also implement comprehensive routines for the preprocessing and integration of the RNA velocity field. Finally, we illustrate the ability of our method to recover differentiation trees with high accuracy on several simulated and real datasets, and compare these results with the state of the art.}




\maketitle

\section{Introduction}\label{sec:introduction}

RNA sequencing is the process of quantifying and identifying RNA molecules in a biological sample. Based on these measurements, it is then possible to estimate what is known as the expression profile of the sample, a vector whose entries encode the expression levels of reference genes in the sample~\cite{kukurba_rna_2015}. Single-cell RNA sequencing (scRNA-seq) enables such vector to be estimated at the scale of each individual cell. Variation in gene expression across a population of cells reflects an underlying dynamic biological process, such as development, or disease progression. It is this process that we are interested in characterizing. 

\subsection{Trajectory inference, a tree inference problem} 
RNA sequencing provides only a snapshot of the cell population at a specific moment in time, and the stage of the process at which each cell was observed at that moment is unknown. To reconstruct the process, one must therefore order the cells along a developmental or temporal trajectory based on similarities in their expression profiles. This is known as \textit{trajectory inference}~\cite{saelens_comparison_2019}. A common assumption in the literature is that the underlying biological process is cellular differentiation, a process during which cells increasingly specialize. In this case, the process can be represented as a tree in which each node corresponds to a stage of the process, and the trajectory inference problem is precisely that of inferring such a differentiation tree from the sequencing data. Most trajectory inference methods specifically designed to infer differentiation trees propose reconstructing the target tree as the minimum spanning tree of the k-nearest neighbors graph of the population of cells, as defined by cell similarities~\cite{street_slingshot_2018, qiu_monocle_2021}. But the resulting tree is very sensitive to noise in the observations. In practice, minimum spanning trees are therefore mainly used as estimators for coarser differentiation trees and must be supplemented by finer estimators such as principal curves~\cite{street_slingshot_2018}. Other methods which do not necessarily assume a differentiation process can be used to identify the branches~\cite{weng_vetra_2021, zhang_cellpath_2021} or the leaves~\cite{setty_palantir_2019, lange_cellrank_2022} of the differentiation tree of a population of cells, and to order cells within each branch (or towards each leaf, respectively). However, it is not clear how to aggregate those orderings so as to recover the differentiation tree. Furthermore, some of these methods also require knowing the number of branches in advance~\cite{weng_vetra_2021, zhang_cellpath_2021}. Beyond transcriptomics, tree inference -- the problem of reconstructing a tree from some observations in order to represent the relationships between the observed objects -- is a central problem in phylogenetics. There is therefore a wide range of methods for this purpose, several of which fall under the approach known as \textit{distance-based tree inference}~\cite{pardi_combinatorics_2012}. The goal of this article is to investigate this approach, which, to the best of our knowledge, has not yet been considered in the context of trajectory inference. In distance-based tree inference, one assumes a dissimilarity measure on the observations -- here, the scRNA-seq data, and then searches for a tree such that the path distance between two nodes coincides with the dissimilarity between the observations at these nodes~\cite{sokal_upgma_1958, saitou_nj_1987, lefort_fastme_2015}. We expect distance-based tree inference to be more robust to noise than approaches based on minimum spanning trees, while being able to operate at a finer scale. The key to implementing this approach is then to define a well-behaved dissimilarity measure on scRNA-seq data. 

\subsection{RNA velocity}
Recent advances in scRNA-seq make it possible to distinguish between spliced and unspliced RNA, the ratio of which can then be used to estimate RNA velocity, a vector that encodes the direction and rate of changes in gene expression levels undergone by each cell at the time of observation~\cite{manno_rna_2018, gorin_rna_2022}. Latest scRNA-seq data consist then of a discrete vector field on a subset of a Euclidean space -- of dimension equal to the number of genes under consideration (Fig.~\ref{fig:velocity}).
\begin{figure}[!h]
\centering
    \includegraphics[width=.7\textwidth]{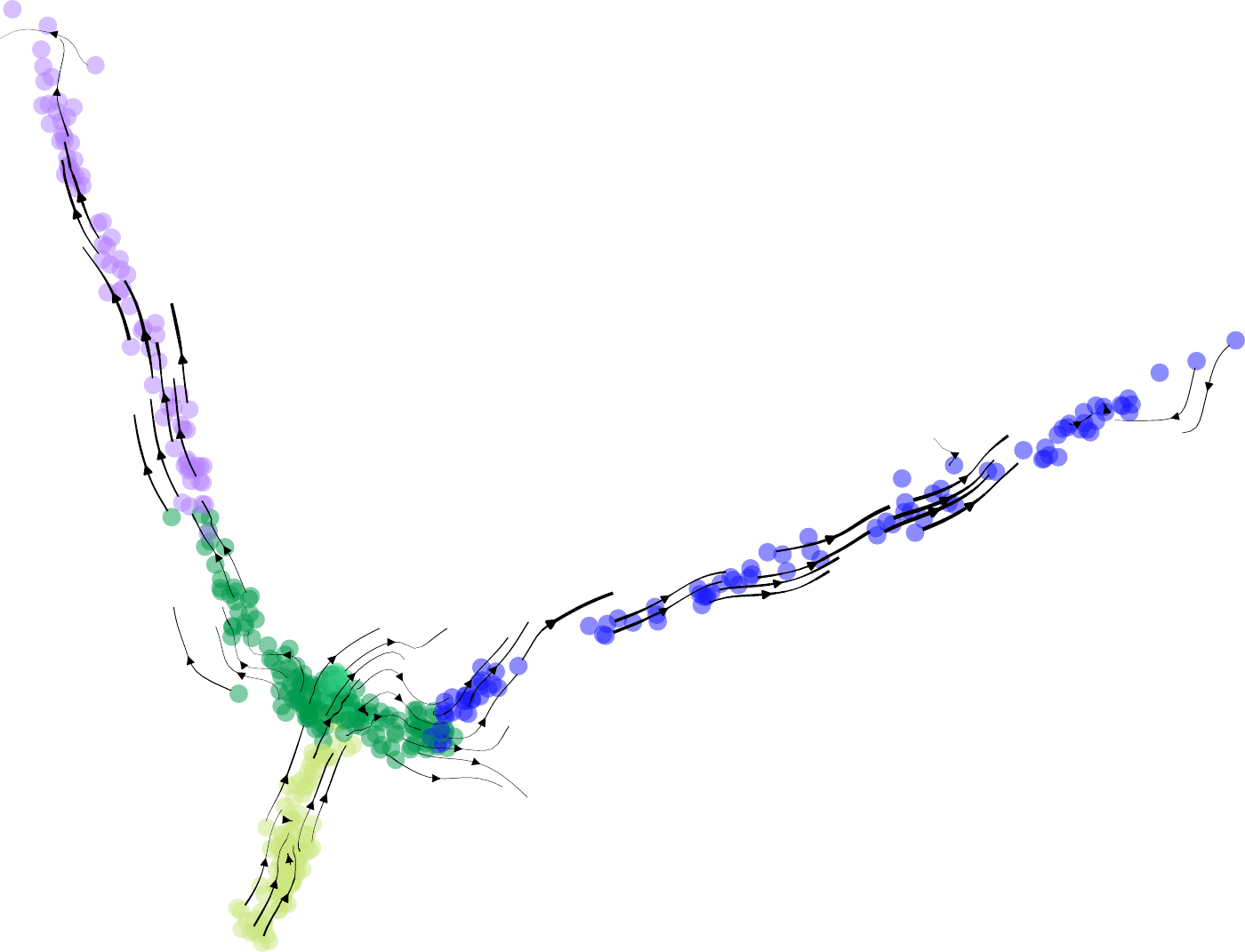}
    \vspace{.6cm}
    \caption{RNA velocity field of a bifurcating trajectory simulated with the \texttt{dyngen} library~\cite{robrecht_dyngen_2021}. The 2D visualization is generated via principal component analysis (PCA), superimposed with the \texttt{streamplot} function of \texttt{matplotlib} for a better rendering of the RNA velocity field. The colors correspond to cell states as labeled by \texttt{dyngen}.}
    \vspace{-.4cm}
    \label{fig:velocity}
\end{figure}

Several trajectory inference methods introduce dissimilarity measures on scRNA-seq data that take into account RNA velocity. A recurring one is a measure of the temporal inconsistency between two neighboring cells based on Pearson correlation coefficient between the RNA velocity of the cell located downstream of the other cell and the displacement vector from the former cell to the latter~\cite{weng_vetra_2021, zhang_cellpath_2021, lange_cellrank_2022}. While such a measure is well suited for gauging the dissimilarity between cells with a temporal relationship, it tends to over-separate cells at the same stage of the process, that is, cells with similar expression profiles and a common RNA velocity orthogonal to the displacement vector from one cell to another. To counterbalance this, VeTra~\cite{weng_vetra_2021} also introduces a second dissimilarity measure, defined as the $L^2$ distance on the product space of expression profiles and RNA velocities. However, we identify a major limitation to using this dissimilarity measure for tree inference, namely that a Euclidean distance cannot accurately capture the hyperbolic nature of the path distance on a tree.

\subsection{Our method}
In this article, we propose a new dissimilarity measure for the comparison of scRNA-seq data, defined as the squared \textit{varifold distance} between the curves obtained by integrating the corresponding RNA velocity field backwards. The varifold distance is a parametric distance introduced by~\citeauthor{kaltenmark_general_2017} for comparing curves from a shape perspective. Intuitively, the integral curves of the RNA velocity field describe for each individual cell the evolution of its gene expression levels from an initial stage to its current stage. Under the assumption that the relative position of such curves in the expression space reflects the hierarchy between the paths from the root to the nodes in the underlying differentiation tree, we show that for well-chosen parameters, the squared varifold distance defines a dissimilarity measure between these curves that is topologically equivalent to the path distance in the target differentiation tree. We develop these results into an end-to-end method for inferring differentiation trees from RNA velocity fields. In particular, we detail processing steps for denoising RNA velocity fields, and combine our dissimilarity measure with a distance-based tree inference method called family-joining~\cite{kalaghatgi_family_2016}. We then apply our method to several synthetic datasets simulated via the \texttt{dyngen} library~\cite{robrecht_dyngen_2021} as well as a RNA sequencing dataset of mouse pancreatic endocrine cells~\cite{bastidasponce_comprehensive_2019}, and we demonstrate the method's ability to recover the underlying differentiation tree with high accuracy. We compare our results with those achieved by two other methods, VeTra~\cite{weng_vetra_2021} and CellPath~\cite{zhang_cellpath_2021}. This article is an extended version of our work~\cite{maignant_tree_2025}, which was limited to the theoretical findings on the dissimilarity measure.

\subsection{Organization of the article}
The article is organized as follows. In Section~\ref{sec:background}, we provide background knowledge on the tools our method relies on, specifically on the varifold distance and on the distance-based tree inference method called family-joining. In Section~\ref{sec:methods}, we describe the workflow of our method and detail the different steps, in particular those related to data processing, as well as how to adjust the different hyperparameters. In Section~\ref{sec:theory}, we derive theoretical guarantees on the accuracy of our method in inferring the true differentiation tree under certain assumptions on the data. Finally, we present in Section~\ref{sec:experiments} several experiments on simulated and real scRNA-seq datasets. 

\section{Background} \label{sec:background}

Our method relies on two existing tools. The first is the varifold distance~\cite{kaltenmark_general_2017} between curves -- in our case the integral curves of the RNA velocity field, which we already presented in~\cite{maignant_tree_2025} and based on which we defined our cell dissimilarity measure. The second is a distance-based tree inference method, family-joining~\cite{kalaghatgi_family_2016}, which allows us to reconstruct the target differentiation tree from the dissimilarity matrix defined. In this section, we give a general understanding of the two tools and motivate these choices for our problem.
    
\subsection{Varifold distance between curves}\label{sec:varifold}
Let us then first recall the definition of the varifold distance. Precisely, we follow the framework by~\citeauthor{kaltenmark_general_2017} known as \textit{oriented varifolds}. The general idea of their approach is to represent curves as distributions over the space of positions and oriented tangent vectors, and then to embed the set of these distributions into a Hilbert space. Via this representation, the space of curves is then equipped with a metric structure and thereby a distance measure, namely the aforementioned varifold distance. The literature proposes numerous other dissimilarity measures between curves, whether in the context of shape analysis or signal processing. The choice of the varifold distance over any other dissimilarity measure is essentially justified by the theoretical guarantees we prove in Section~\ref{sec:theory}. Intuitively, the varifold distance maximizes the separation between divergent curves. Moreover, the varifold distance also offers several other desirable properties, such as being invariant under reparameterization while remaining computationally efficient, and being robust to deformations of the curves under consideration. 
\medbreak  
Let $\gamma$ be a smooth curve of finite length embedded in $\mathbb{R}^n$ and $W$ be a space of test functions on $\mathbb{R}^n \times \mathbb{S}^{n-1}$, the oriented varifold $\mu_\gamma$ associated to $\gamma$ is the element of the dual $W^\ast$ given by 
\begin{equation}
    \mu_\gamma: W \to \mathbb{R}, \;\omega \mapsto \int_\gamma \omega(x,\vec{t}(x)) \;d\ell(x),
\end{equation}
where the unit length vector $\vec{t}(x) \in \mathbb{S}^{n-1}$ encodes the oriented tangent vector to $\gamma$ at $x$. The representation $\mu_\gamma$ depends on the shape and orientation of tangent spaces but not on a choice of a parameterization. Let us also mention the following additivity property: for distinct curves $\gamma_1$ and $\gamma_2$, we have that $\mu_{\gamma_1 \cup \gamma_2} = \mu_{\gamma_1} + \mu_{\gamma_2}$. This property will be useful for our subsequent theoretical results.
\medbreak    
Distances on $W^\ast$ and thus between curves are induced from a \textit{Reproducing Kernel Hilbert Space} structure on the set of test functions $W$ defined by a positive definite kernel $k : (\bR^n \times \mathbb{S}^{n-1}) \times (\bR^n \times \mathbb{S}^{n-1}) \to \mathbb{R}$. In particular, the reproducing property guarantees that evaluation functionals $\delta_{(x,t)}: \omega \mapsto \omega(x,t)$ can be represented in terms of the inner product 
\begin{equation}
    \delta_{(x,t)}(\omega) = \langle k_{(x,t)}, \omega \rangle_W, 
    \vspace{.3cm}
\end{equation}
where $k_{(x,t)}(\cdot) = k((x,t),\cdot)$ is called the reproducing kernel. This allows us to define an inner product on the dual by 
\begin{equation}
    \langle \delta_{(x_1,t_1)}, \delta_{(x_2,t_2)} \rangle_{W^\ast} = \langle k_{(x_1,t_1)}, k_{(x_2,t_2)}\rangle_W = k((x_1,t_1),(x_2,t_2)) ,
\end{equation}
where the last equality follows from the construction of $W$ from $k$. While there is a wide range of feasible kernels, we restrict ourselves to tensor products of Gaussian kernels $k_\sigma(x,y)=e^{-\|x - y\|^2/\sigma^2}$ that are of the form $k((x_1,t_1), (x_2,t_2))=k_{\sigma_x}(x_1,x_2)k_{\sigma_t}(t_1,t_2)$. Then, we write the inner product for oriented varifolds as  
\begin{align}
\langle \mu_{\gamma_1}, \mu_{\gamma_2}\rangle_{W^\ast} &= \Bigl\langle \int_{\gamma_1} \delta_{(x,\vec{t}(x))} d\ell(x), \int_{\gamma_2} \delta_{(y,\vec{t}(y))} d\ell(y)\Bigr\rangle_{W^\ast} \\
&= \iint_{\gamma_1 \times \gamma_2} e^{-\frac{\|x - y\|^2}{\sigma_x^2}} e^{-\frac{\|\vec{t}(x) - \vec{t}(y)\|^2}{\sigma_t^2}} \;d\ell(x)d\ell(y).
\end{align}
The separable structure enhances the interpretability of contributions from spatial and orientation characteristics. Intuitively, $k_{\sigma_x}$ measures proximity between point positions, whereas $k_{\sigma_t}$ quantifies the proximity between the associated tangent spaces. For our choice of kernel, the identification $\gamma \to \mu_\gamma$ is injective such that the restriction of the metric on $W^\ast$ defines a distance on the space of curves:
\begin{equation}
    d_{W^\ast}(\gamma_1, \gamma_2) = \|\mu_{\gamma_1} - \mu_{\gamma_2}\|_{W^\ast}.
\end{equation}
This distance can be efficiently approximated due to closed-form expressions for discrete curves such as polygonal chains and is insensitive to the sampling of discrete curves due to the parameterization invariance. Finally, it is robust with respect to small diffeomorphic deformations. Precisely, for two smooth curves $\gamma_1$ and $\gamma_2$ we have
\begin{equation} \label{eq:robust}
    d^2(\gamma_1, \phi.\gamma_2) \underset{\substack{\vspace{.05cm} \\ \|\phi - id\|_\infty \to 0 \vspace{.05cm} \\ \|d\phi - I\|_F \to 0}}{\to} d^2(\gamma_1, \gamma_2).
\end{equation}
Similar properties have been proven in related literature~\cite{glaunes_these_2005, charlier_fshape_2017}. In practice, for the distance between $\gamma_1$ and $\gamma_2$ to be robust to a deformation $\phi$ of $\gamma_2$, we need to choose $\sigma_x$ significantly greater than $\|\phi - id\|_\infty$ and $\sigma_t$ significantly greater than $\|d\phi - I\|_F$.

\subsection{Distance-based tree inference}\label{sec:fj}
Distance-based tree inference algorithms were initially developed in phylogenetics to reconstruct evolutionary relationships between species, for example using genomic data. Most of these algorithms therefore assume that observations lie only at the leaf nodes of the target tree, in accordance with the hypothesis that the most recent common ancestor of two extant species is always an extinct species and therefore cannot be observed. In the context of transcriptomics, however, we cannot make this kind of hypothesis. Indeed, each cell in a sample evolves at a very different rate, and the typical time of observation is generally less than the duration of the process, so that the cells observed may still be in the early stages of their differentiation. Consequently, we expect our observations to lie at all possible nodes of the target differentiation tree, and not just at the leaf nodes. One algorithm designed to handle this scenario is the method called family-joining~\cite{kalaghatgi_family_2016}. Family-joining is a variant of the better-known neighbor-joining heuristic~\cite{saitou_nj_1987}, which consists of iteratively grouping observations into pairs of adjacent leaves according to their similarity. Family-joining relaxes the requirement that observations be located at leaf nodes, and additionally determines at each iteration whether the two most similar observations should be linked as sibling nodes -- and if so, which observation is their parent -- or as parent and child nodes. This decision is based on the following principle: for two observations $O_i$ and $O_j$ to be in a parent-child relationship, the dissimilarity $\Delta_{ik}$ between the parent observation $O_i$ and any other third observation $O_k$ must be equal to the sum $\Delta_{ij}+\Delta_{jk}$ of the dissimilarity between the parent observation and the child observation, and the dissimilarity between the child observation and the third observation. We implement the pseudocode provided along with this algorithm in~\cite{kalaghatgi_family_2016}. Please note that the original algorithm also involves introducing an unobserved node as the parent of two sibling nodes in cases where no observation would be suitable. In our implementation, we exclude this option and systematically assign the best candidate observation to the parent node of the two sibling nodes. Indeed, given the large number of observations available in our context, we can assume that there are no unobserved nodes. In this way, we eventually reconstruct the smallest tree possible. 

\section{Methods}\label{sec:methods}
Building on the tree inference tools described in the previous section, our method incorporates two initial steps: preprocessing the velocity field, and integrating it. Below, we present the full pipeline and detail how to select or tune the parameters involved in each step of the algorithm.

\subsection{Algorithm}
The input data consists of the expression profiles $x_1, \dots, x_N \in \bR^n$ and the RNA velocities $v_1, \dots, v_N \in \bR^n$ of $N$ cells. Upstream of the analysis, we reduce the dimension of the data from $n$ to $d < 30$ in order to retain only the significant variability and reduce the computational load. Our method then comprises four main steps:
\vspace{-.2cm}
\begin{enumerate}[label={}]
    \item\textbf{Step\;1.\;Denoise the RNA velocity field} by first smoothing it and then projecting it back onto the local tangent spaces to the point cloud described by the expression profiles. 
    \medbreak
    \item\textbf{Step\;2.\;Integrate the RNA velocity field} backward and recover for each cell the integral curve $\gamma_i$ defined as the solution of the differential equation
    \begin{equation}
       \dot\gamma(t) = - \sum_{i=1}^N\frac{K(\gamma(t), x_i)}{\sum_{j=1}^N K(\gamma(t), x_j)}v_i
    \end{equation}
    for the initial condition $\gamma_i(0) = x_i$ and where $K$ is a smoothing kernel that allows to derive a smooth velocity field on $\bR^d$ from the RNA velocity field.
    \medbreak
    \item\textbf{Step\;3.\;Compute the dissimilarity matrix} defined as the squared varifold distance between the integral curves obtained at step 2:
    \begin{equation}
        \Delta_{ij} =  d_{W^\ast}(\gamma_{i}, \gamma_{j})^2.
    \end{equation}
    \medbreak
    \item\textbf{Step\;4.\;Infer the differentiation tree} from $\Delta$ using family-joining (see Section~\ref{sec:fj}). 
\end{enumerate}

\subsection{Dimensionality reduction}
The raw expression profiles $x_1, \dots, x_N \in \mathbb{R}^n$ and the RNA velocities $v_1, \dots, v_N \in \mathbb{R}^n$ are projected onto the first $d$ principal components of the $n$-dimensional point cloud described by $x_1, \dots, x_N$. The number $d$ of components is an approximate estimate of the number of significant components in terms of the fraction of the variance in $x_1, \dots, x_N$ they explain (as defined by PCA). Note that we denote indistinctly by  $x_1, \dots, x_N \in \mathbb{R}^d$ and $v_1, \dots, v_N \in \mathbb{R}^d$ the data after dimensionality reduction.

\subsection{Preprocessing of the RNA velocity field}
RNA velocity estimation is highly sensitive: on the one hand, technical variability, dropouts, and ambient RNA can distort spliced/unspliced counts that are the basis for velocity models; on the other hand, kinetic assumptions such as constant transcription, splicing, and degradation rates are frequently violated in heterogeneous or branching systems. We therefore perform a denoising preprocessing step.

\subsubsection{Smoothing}
Ideally, we would like to improve the signal-to-noise ratio by smoothing the velocity field, that is averaging the velocities of cells that are at the same stage of the differentiation process. In lieu of such information -- the trajectory information, we formulate the smoothing in terms of diffusion processes that integrate the local geometry of the point cloud $x_1, \dots, x_N \in \bR^d$.
Precisely, we employ the heat kernel construction from the popular \textit{diffusion maps}~\cite{coifman_diffusion_2026} method, which recovers the geometry of the data without assuming the samples to be uniformly distributed. A critical parameter in the construction is the diffusion time $t_{\text{diff}}$ that can be hard to tune. We determine $t_{\text{diff}}$ using the semigroup criterion proposed by \citeauthor{shan_diffusion_2022}. A smoothed velocity field is then obtained by applying the heat operator to $v_1, \dots, v_N \in \bR^d$ (Fig.~\ref{fig:smoothed}).
\begin{figure}[!h]
\centering
    \includegraphics[width=.7\textwidth]{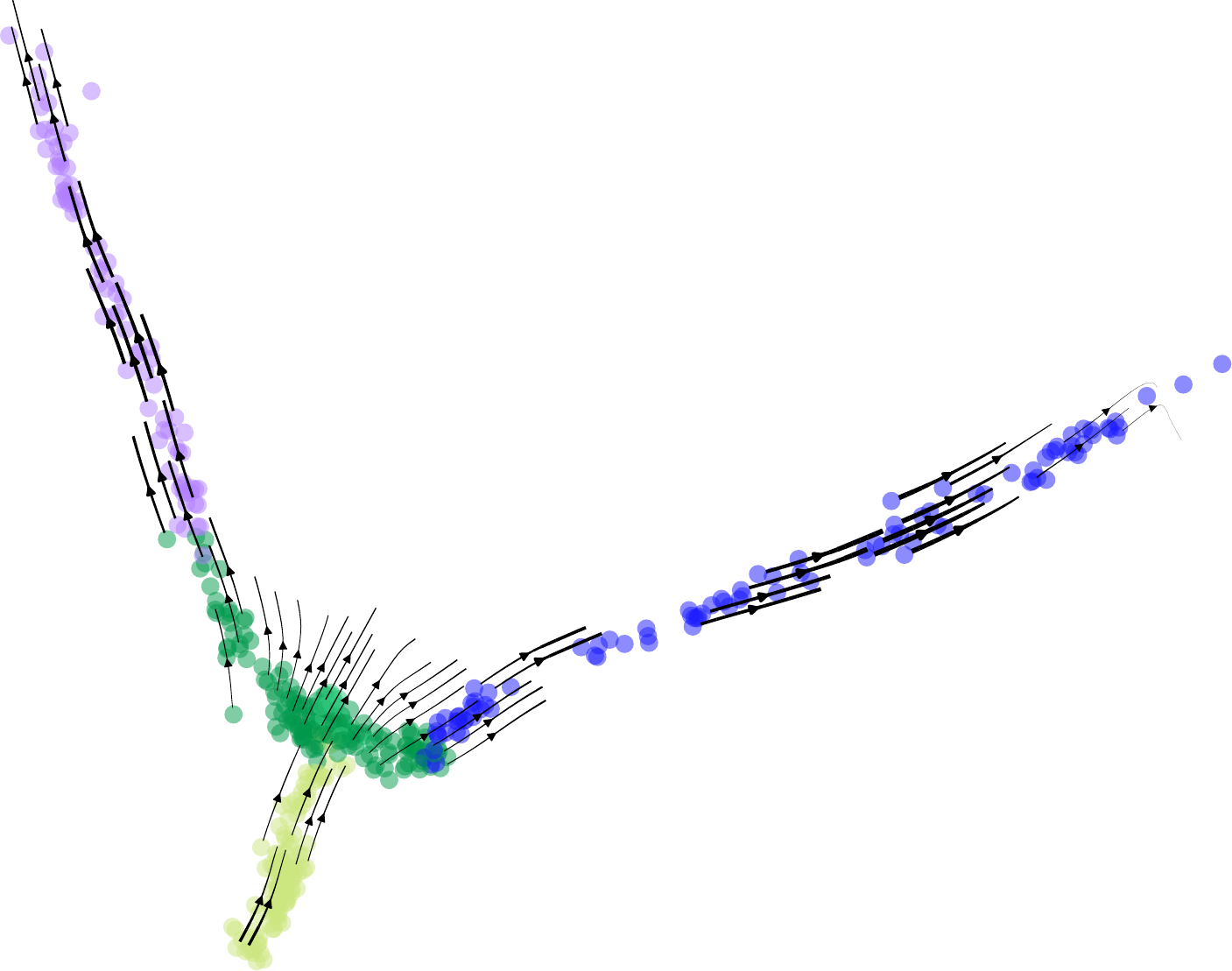}
    \vspace{.5cm}
    \caption{Smoothed RNA velocity field of the bifurcating trajectory pictured in Fig.~\ref{fig:velocity}.}\label{fig:smoothed}
\end{figure}

\subsubsection{Tangent projection}\label{sec:projection}
We further improve the consistency of the RNA velocities by projecting them onto the tangent bundle to the point cloud $x_1, \dots, x_N \in \bR^d$ (Fig.~\ref{fig:projected}). Given a radius $r$, for each cell $i$ we estimate the tangent space at $x_i$ as the span of the principal components derived from the covariance of the $r$-ball neighborhood of $x_i$. The number $d_i$ of components used, hence the dimension of the estimated tangent space at $x_i$, is determined by the maximal variance gap, viz. $d_i=\operatorname{argmax} \lambda_d / \lambda_{d+1}$, following the same criterion as in \cite{ahn_eigenvalue_2013}. The radius $r$ is chosen small enough to capture the local geometry of the point cloud  $x_1, \dots, x_N \in \bR^d$ while exceeding the noise level. In the case of differentiation trajectories, we expect the point cloud $x_1, \dots, x_N \in \bR^d$ to be locally $1$-dimensional. In practice, we do observe that beyond a certain radius -- which corresponds to the noise level, the average local dimension decreases as the radius of the neighborhoods increases, reaching values close to $1$, before increasing again (as we exit the local scale and begin to detect the global geometry of the point cloud). It is the intermediate regime that interests us, and we thus choose the radius $r$ that minimizes the average local dimension. To improve the robustness, the neighborhoods do not include the cells they are constructed for and are further augmented by nearest neighbors of distance larger than $r$ until a minimal cardinality $k$ is reached. Such a threshold $k$ is chosen so as to optimize the convergence of the integral curves of the velocity field, meaning to minimize the distance between those at the final integration step.
\begin{figure}[!h]
\centering
    \includegraphics[width=.7\textwidth]{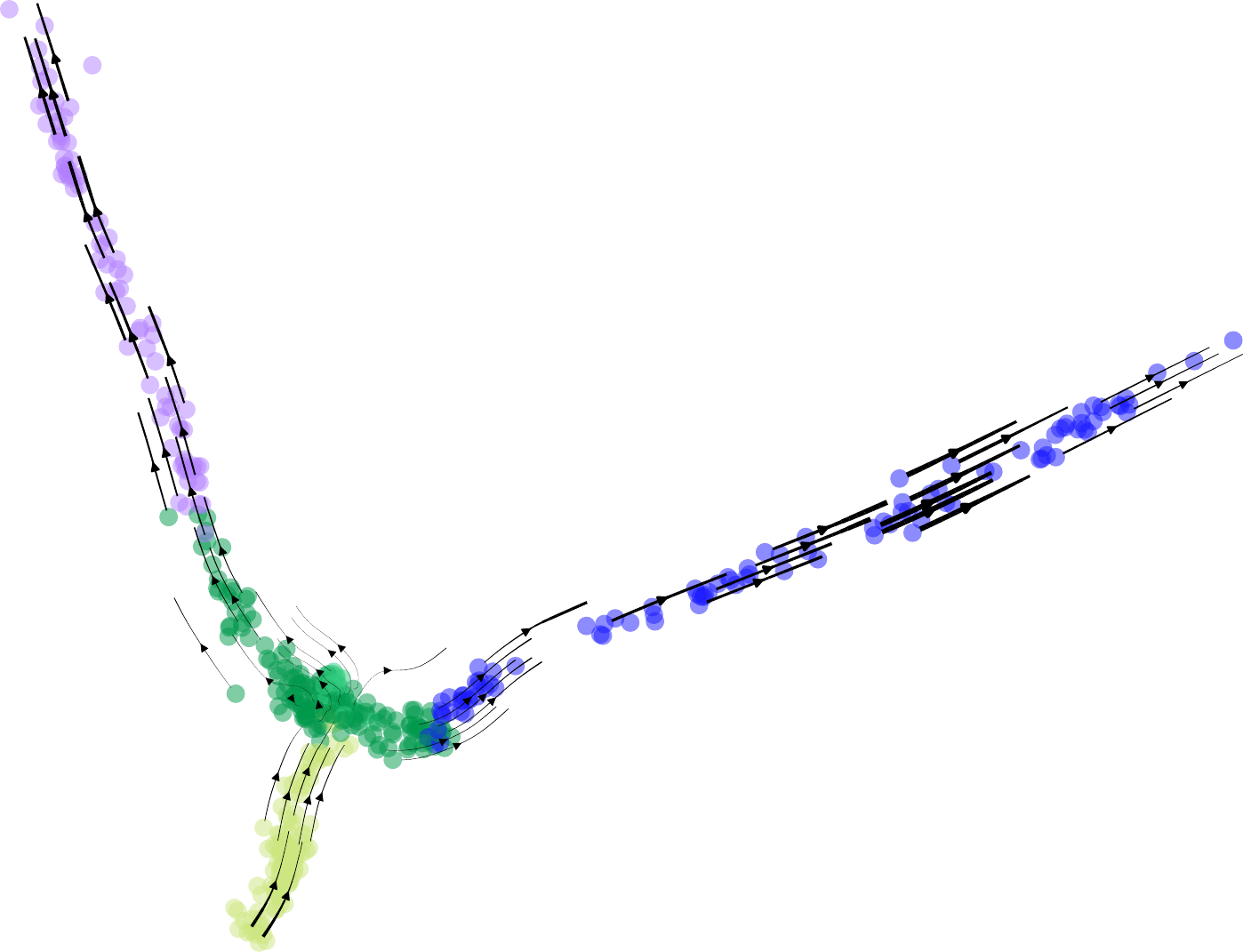}
    \vspace{.6cm}
    \caption{RNA velocity field of the bifurcating trajectory after projection.}\label{fig:projected}
    \vspace{-.9cm}
\end{figure}

\subsection{Integrating the RNA velocity field}
We apply kernel regression to the preprocessed RNA velocity field \cite{nadaraya_estimating_1964, watson_smooth_1964} in order to recover a smooth velocity field defined at every point in $\bR^d$, which can then be integrated (Fig.~\ref{fig:integration}). For the choice of the smoothing kernel $K$, we opt for a Gaussian kernel with a bandwidth equal to $r/3$, so that the typical smoothing scale is the same as the size of the neighborhoods estimated in the previous step. The integration routine we implement requires setting three additional parameters: the integration step size, the discretization step size, and the integration “margin”. The discretization step determines the number of time points output for each integral curve of the velocity field. It is chosen such that the longest integral curve consists of approximately 100 time points -- sampled uniformly. The number of time points represents a trade-off between accurately describing the shape of the integral curves and maintaining a reasonable runtime for computing the dissimilarity matrix. The integration step size is then chosen to be 5 times smaller than the discretization step size. The integration margin corresponds to the distance between the boundary of the integration domain and the point cloud $x_1, \dots, x_N \in \bR^d$. Exiting the integration domain serves as a stopping criterion for the integration. The integration margin is chosen large enough so that most integral curves would have converged at the final integration step, which should be the case asymptotically assuming that all cells observed in the dataset originate from a common root cell.
\begin{figure}[!h]
\centering
    \includegraphics[width=.7\textwidth]{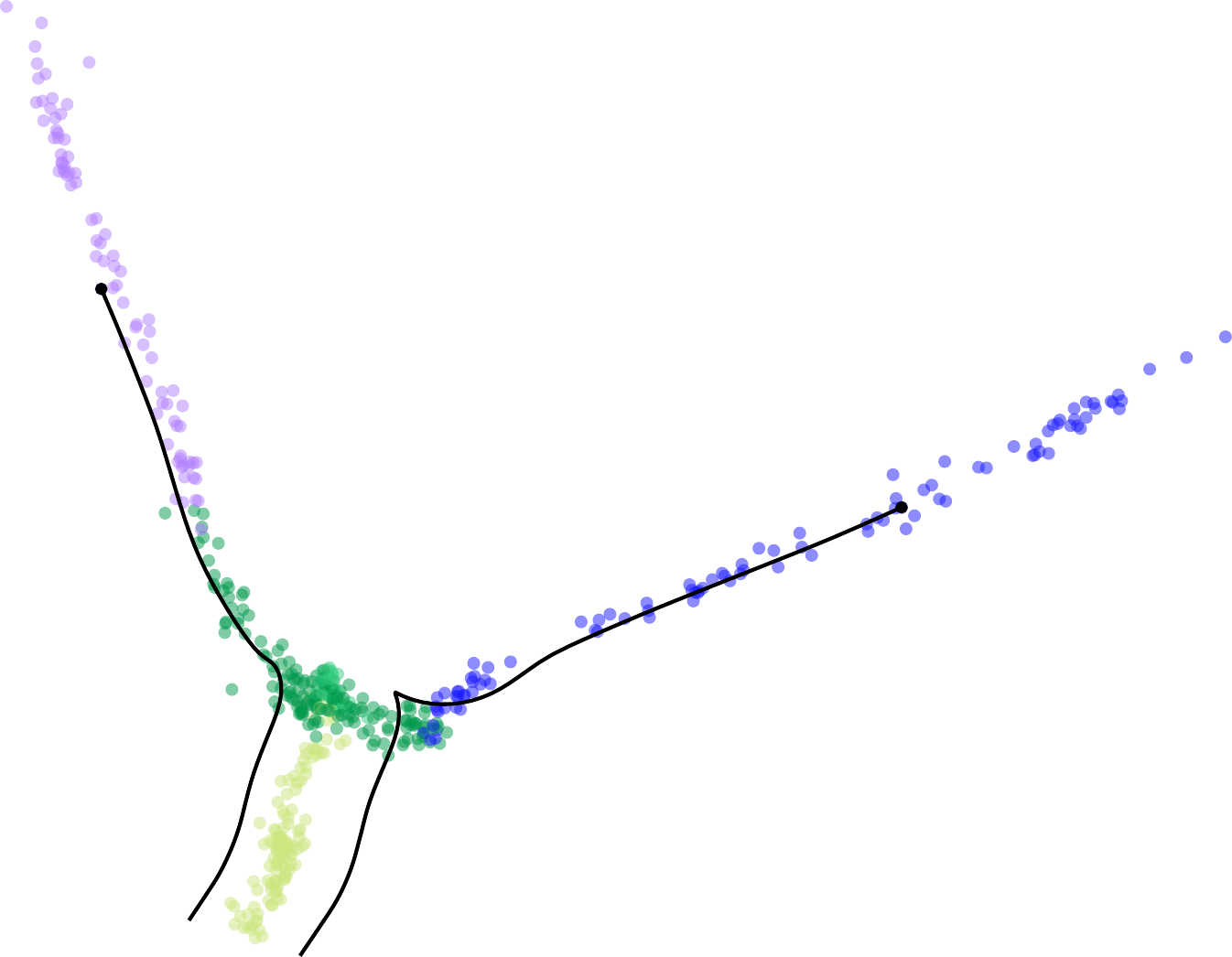}
    \vspace{.6cm}
    \caption{Two branching integral curves of the velocity field of the bifurcating trajectory.}\label{fig:integration}
    \vspace{-.9cm}
\end{figure}

\subsection{Parameter tuning for the dissimilarity matrix}
There are two parameters to tune for the dissimilarity matrix: the spatial sensitivity $\sigma_x$ and the angular sensitivity $\sigma_t$ of the varifold distance. We set the former to a value large enough for the integral curves of the velocity field to be perceived as overlapping during the latest integration steps, specifically to $10$ times the $L^2$ distance between the two furthest curves at the final integration step. As for the angular sensitivity, there is a trade-off between the dissimilarity measure being robust to noise on the integral curves, which requires that the angular sensitivity dominates the angular noise as described by Equation \ref{eq:robust}, and the accurate detection of bifurcation points, which, as we show in Section \ref{sec:theory}, is the case when the angular sensitivity -- and the spatial sensitivity -- tends to $0$. To this end, we estimate the average maximal $L^2$ distance between two branching integral curves (Fig.~\ref{fig:integration}) and that between two non-branching integral curves, and we set $\sigma_t$ such that the corresponding Gaussian kernel best separates the two average distances.

\subsection{Parameter selection for the tree reconstruction}
Family-joining requires defining a threshold above which the criterion that determines whether two nodes are in a parent-child relationship (see Section \ref{sec:fj}) is considered satisfied. We simply set this threshold systematically to a tenth of the average computed dissimilarity. Intuitively, a threshold that is too high results in missing branches, while a threshold that is too low leads to excessive branching. 

\section{Theoretical guarantees}\label{sec:theory}

In this section, we consider a simplified version of our tree inference problem, in which we assume a tree that is faithfully embedded in $\bR^d$, which we observe as the curves representing the paths from the root to the nodes. By placing ourselves in this ideal case, we derive theoretical guarantees on the property of the dissimilarity matrix $\Delta$ to characterize the target differentiation tree.

\subsection{Assumptions}
Let $T = (V, E, r)$ be a rooted tree embedded into $\bR^d$ via the two injective maps
\begin{equation}
    \begin{array}{ccl}
        i \in V \mapsto x_i \in \bR^n
    \end{array} 
    \qquad \text{and} \qquad
    \begin{array}{ccl}
        (i, j) \in E \mapsto \gamma_{x_i, x_j} \subset \bR^d,
    \end{array} 
\end{equation}
where $\gamma_{x_i, x_j}$ denotes a smooth curve of $\bR^d$ connecting $x_i$ and $x_j$ (Fig.~\ref{fig:tree}). 
\begin{figure}[t]
\centering
\includegraphics[width=.7\textwidth]{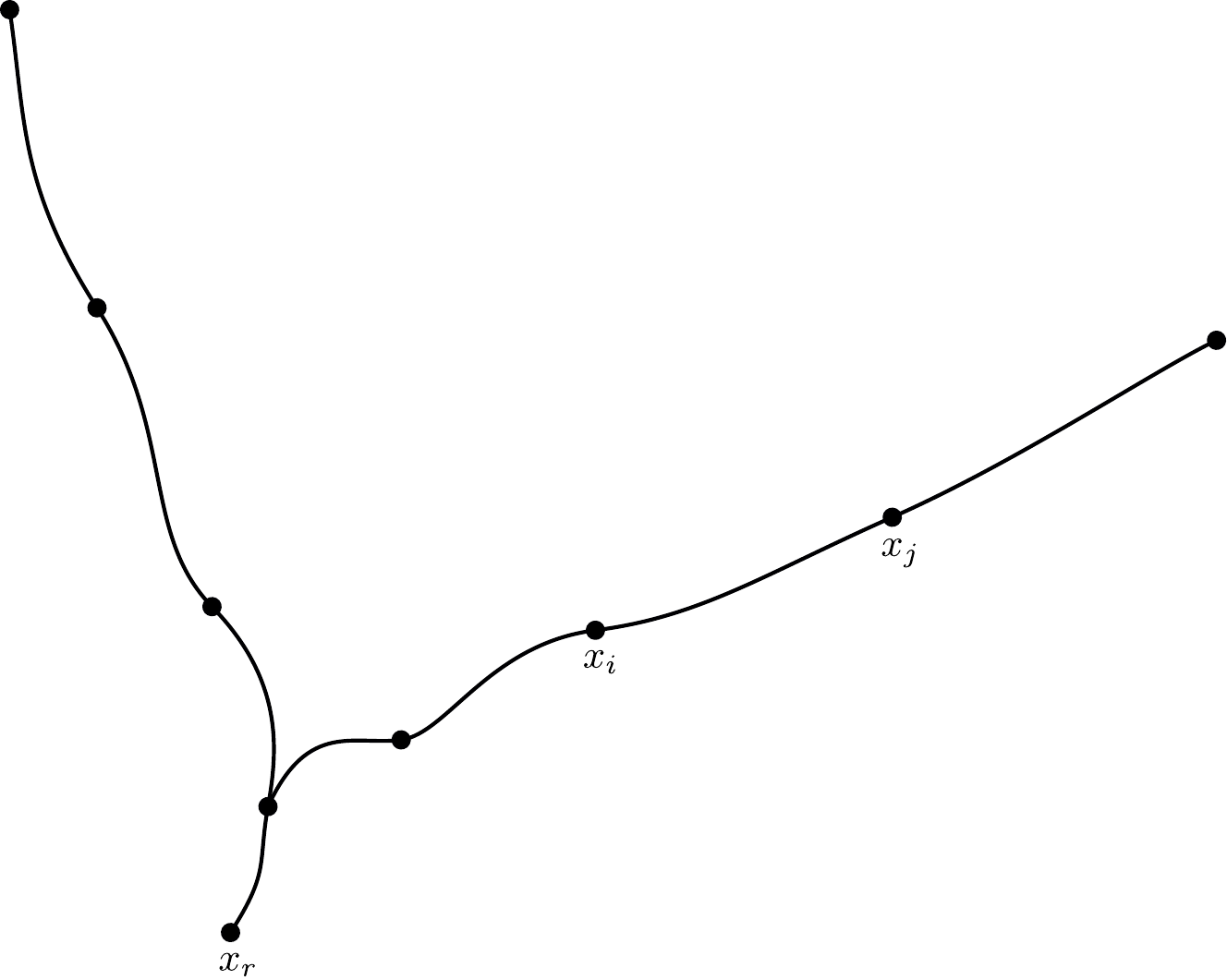}
\vspace{.6cm}
\caption{A rooted tree embedded in the plane.}\label{fig:tree}
\vspace{-.3cm}
\end{figure}
We ask for the embedding to preserve the topology of the tree, that is for the respective embeddings of two edges to intersect only if the edges share a common node, and only at the image of this common node. Additionally, we require that the embedding of the path from the root $r$ to any node is also a smooth curve, and we extend the notation $\gamma_{x, y}$ to the curve segment joining any two points $x$ and $y$ on the embedding of a common path from the root to some node. In particular, the map
\begin{equation}
    i \in V \to \gamma_i \triangleq \gamma_{x_r, x_i} \subset \bR^d
\end{equation}
should denote the embedding of the path from the root $r$ to the node $i$. Finally, we assume that the embedding of any such path is a regular and orientable curve, so that it defines an oriented varifold~\cite{kaltenmark_general_2017}. A natural orientation would be, for example, from the root to the end node. We are then interested in reconstructing $T$ based on dissimilarities between the curves $\gamma_1, \dots, \gamma_N$. To ensure that the $d$-dimensional embedding of $T$ reflects its topology well enough, we make three further assumptions.
\vspace{-.4cm}
\begin{assumption}
\label{A1} 
The distance between two adjacent edges of the same path is increasing as they move further from their common node. Let $i, j, k$ be three nodes of $T$ such that $i < j < k$ meaning $i$ is the parent node of $j$ which is itself the parent node of $k$. Then for all $x, x' \in \gamma_{x_i, x_j}$ and $y, y'\in \gamma_{x_j, x_k}$ we have that
\begin{equation*}
    \ell(\gamma_{x, y}) \leq \ell(\gamma_{x', y'}) \Rightarrow \|x - y\| \leq \|x' - y'\|
\end{equation*}   
where $\ell(\gamma)$ denotes the arc length of the curve $\gamma$. \\
\begin{figure}[!h]  
    \centering
    \vspace{-.3cm}
    \includegraphics[width=.7\textwidth]{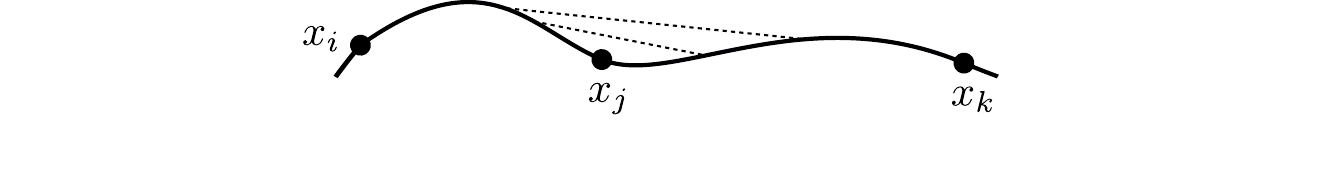} 
    \vspace{-1.4cm}
    \label{fig:hyp1}
\end{figure}
\end{assumption}

\begin{assumption}
\label{A2}
The angle between two adjacent edges of branching paths is increasing as they move further from their common node. Let $i, j, k$ be three nodes of $T$ such that $i < j$ and $i<k$. Then for all $x, x' \in \gamma_{x_i, x_j}$ and $y, y'\in \gamma_{x_i, x_k}$ we have that
\begin{equation*}
\ell(\gamma_{x_i, x}) +  \ell(\gamma_{x_i, y}) \leq \ell(\gamma_{x_i, x'}) + \ell(\gamma_{x_i, y'})\Rightarrow \|\vec{t}(x) - \vec{t}(y)\| \leq \|\vec{t}(x') - \vec{t}(y')\|
\end{equation*}  
where we recall that $\vec{t}(\cdot)$ denotes the unit oriented tangent vector to the curve it is evaluated~on.
\end{assumption}
\begin{figure}[!h]
    \centering
    \vspace{-.6cm}
    \includegraphics[width=.7\textwidth]{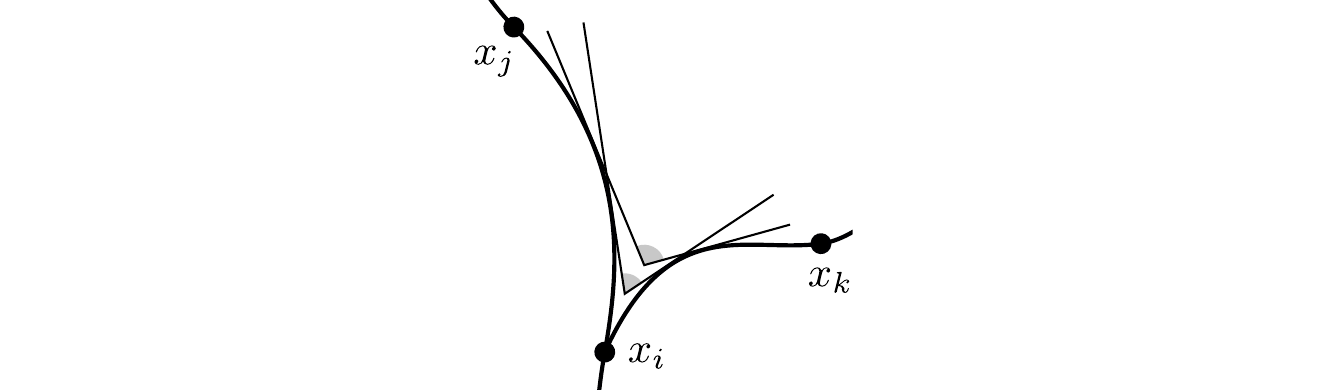}
    \vspace{-.8cm}
    \label{fig:hyp2}
\end{figure}

\begin{assumption}
\label{A3}
Branching paths deviate sufficiently fast from their initial direction. Let $i, j, k$ be three nodes of $T$ such that $i < j$ and $i<k$. Then there exists some neighborhood $U_i$ of $x_i$ and $0 < a < 2$ such that for all $x \in U_i \cap \gamma_{x_i, x_j}$ and all $x \in U_i\cap\gamma_{x_i, x_k}$ we have
\begin{equation*}
    \|\vec{t}(x) - \vec{t}(x_i)\| \geq \ell(\gamma_{x_i, x})^a.
\end{equation*}
\end{assumption}
The first two assumptions are quite natural and essentially state that we can infer from the embedding whether three consecutive nodes are aligned along a common path originating from the root or not. The last assumption is stronger than the second, primarily technical, and guarantees that we can accurately identify branch nodes on the basis of the tangent vectors alone. Ultimately, the three assumptions are only really necessary when the data are noisy and the tree looks more like Fig.~\ref{fig:integration} than Fig.~\ref{fig:tree}.

\subsection{Results}

The dissimilarity measure between the nodes of the tree $T$ is defined as
\begin{equation*}
    \Delta_{ij} = d_{W^\ast}(\gamma_i, \gamma_j)^2
\end{equation*}
    where we recall that $\gamma_i$ and $\gamma_j$ denote the embedding in $\bR^n$ of the paths in $T$ from the root $r$ to nodes $i$ and $j$ respectively. Then $\Delta$ approximates the shortest-path distance on a weighted tree $T'$ that is isomorphic to $T$ in the following sense:
\begin{proposition}
    \label{prop:shortest_path_distance}
    Let $i, j$ be two nodes of $T$. Let $i=i_1,\dots,i_N=j$ denote the shortest path from $i$ to $j$. Then we have the equivalence
\begin{equation}
    \Delta_{ij} \underset{\substack{\vspace{.05cm} \\ \sigma_x, \sigma_t \to 0 \vspace{.1cm} \\ \sigma_x \asymp \sigma_t}}{\sim} \sum_{k=1}^{N-1}\Delta_{i_k i_{k+1}}
\end{equation}
    where $\sigma_x \asymp \sigma_t$ means here that there exist $k_1, k_2 > 0$ such that $k_1 \sigma_x \leq \sigma_t \leq k_2 \sigma_x$.
\end{proposition}
    \noindent Note that this last condition is essentially technical and allows us to control the convergence of $\sigma_x$ and $\sigma_t$ - acting at two different scales - simultaneously. The result is a direct consequence of the two following lemmas.
\begin{lemma}\label{lemma:1}
    Let $i, j$ and $k$ be three adjacent nodes of $T$ such that $i < j < k$. Assume that Assumption \ref{A1} holds. Then we have
\begin{equation}
    \label{eq:L1}
    \langle \mu_{\gamma_{x_i, x_j}}, \mu_{\gamma_{x_j, x_k}}\rangle_{W^\ast} \underset{\substack{\vspace{.05cm} \\ \sigma_x, \sigma_t \to 0 \vspace{.1cm} \\ \sigma_x \asymp \sigma_t}}{=} o\left(\|\mu_{\gamma_{x_i, x_j}}\|_{W^\ast}^2 + \|\mu_{\gamma_{x_j, x_k}}\|_{W^\ast}^2\right).
\end{equation}
\end{lemma}
\begin{proof}
    See Appendix~\ref{app:proofs}.
\end{proof}
\begin{lemma}\label{lemma:2}
    Let $i, j$ and $k$ be three adjacent nodes of $T$ such that $i < j$ and $i < k$. Assume that Assumptions \ref{A2} and \ref{A3} hold. Then we have
\begin{equation}
    \label{L2}
    \langle \mu_{\gamma_{x_i, x_j}}, \mu_{\gamma_{x_i, x_k}}\rangle_{W^\ast} \underset{\substack{\vspace{.05cm} \\ \sigma_x, \sigma_t \to 0 \vspace{.1cm} \\ \sigma_x \asymp \sigma_t}}{=} o\big(\|\mu_{\gamma_{x_i, x_j}}\|_{W^\ast}^2 + \|\mu_{\gamma_{x_i, x_k}}\|_{W^\ast}^2\big).
\end{equation}
\end{lemma}
\begin{proof}
    See Appendix~\ref{app:proofs}.
\end{proof}

Let us sketch out the proof of Proposition \ref{prop:shortest_path_distance}. First, we write the varifold distance between the paths (in $\bR^n$) from the root to $i$ and $j$ as the distance between the two paths from their common ancestor, which is one of the nodes visited by the shortest path in $T$. Then we write each of the two paths as the union of the edges composing them and decompose their squared distance accordingly. For all scalar products thus generated, we apply Lemma 1, unless one of the nodes considered is the common ancestor and a branching point, in which case we apply Lemma 2. 
\medbreak
As a consequence of Proposition \ref{prop:shortest_path_distance}, we also have that $\Delta$ asymptotically satisfies the triangular inequality
\begin{equation}
\Delta_{ik} \underset{\substack{\vspace{.05cm} \\ \sigma_x \to 0 \vspace{.05cm} \\ \sigma_t \to 0}}{\leq} \Delta_{ij} + \Delta_{jk} + o(\sup_{e \in E}\Delta(e))
\end{equation}
and the four-points condition for a distance matrix to be a shortest-path distance on a tree~\cite{pereira_note_1969}
\begin{equation}
\Delta_{ik} + \Delta_{jl} \leq \max(\Delta_{ij} + \Delta_{kl}, \Delta_{jk} +  \Delta_{li}) + o(\sup_{e \in E}\Delta_e).
\end{equation}

\section{Experiments}\label{sec:experiments}

We now evaluate our method on both simulated and real data sets sampling differentiation trajectories. For the simulated datasets, we compare the trajectory inferred by our method to that inferred by VeTra \cite{weng_vetra_2021} and CellPath \cite{zhang_cellpath_2021}, using the ground-truth trajectory that can be derived directly from the simulation model.

\subsection{Simulated comparative study}
We first evaluate our method on several synthetic datasets simulated via the \texttt{dyngen} reference library~\cite{robrecht_dyngen_2021}. Each dataset in \texttt{dyngen} is modeled on what is called the \textit{backbone} of the dataset -- a coarse-grained trajectory consisting of only the root node, branch nodes, and leaf nodes. We consider three different backbones of increasing complexity: a tree with a single branch node of degree $3$ (bifurcating trajectory), already illustrated earlier in this article, a tree with a single branch node of degree $4$ (trifurcating trajectory), and a tree with two consecutive branch nodes of degree $3$ (double bifurcating trajectory). For these three datasets, we verify that the tree inferred by our method closely matches the corresponding backbone. Furthermore, the dissimilarity matrix calculated prior to tree inference itself provides a very faithful representation of the backbone (Fig.~\ref{fig:simulated}). Additionally, we observe that the trees inferred are fairly robust to the choice of parameters for the dissimilarity matrix (reported in Table~\ref{tab:parameters}).
\begin{figure}[!h]
    \centering
    \begin{tabular}{cccc}
    \rotatebox[origin=c]{90}{\small Step 3} &
    \begin{subfigure}[!h]{.28\linewidth}
        \centering
        \includegraphics[width=\linewidth]{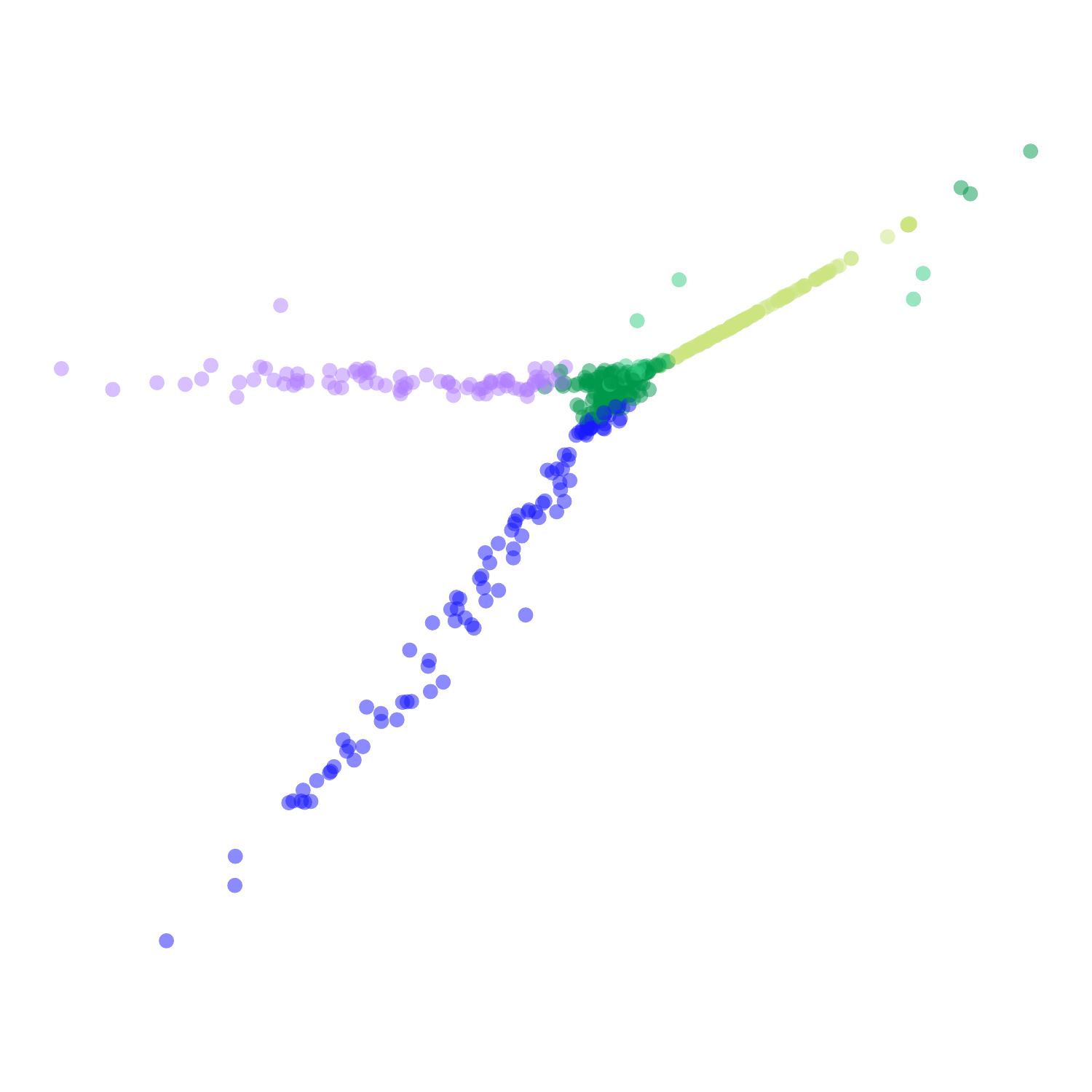}
    \end{subfigure} &
    \begin{subfigure}[!h]{.28\linewidth}
        \centering
        \includegraphics[width=\linewidth]{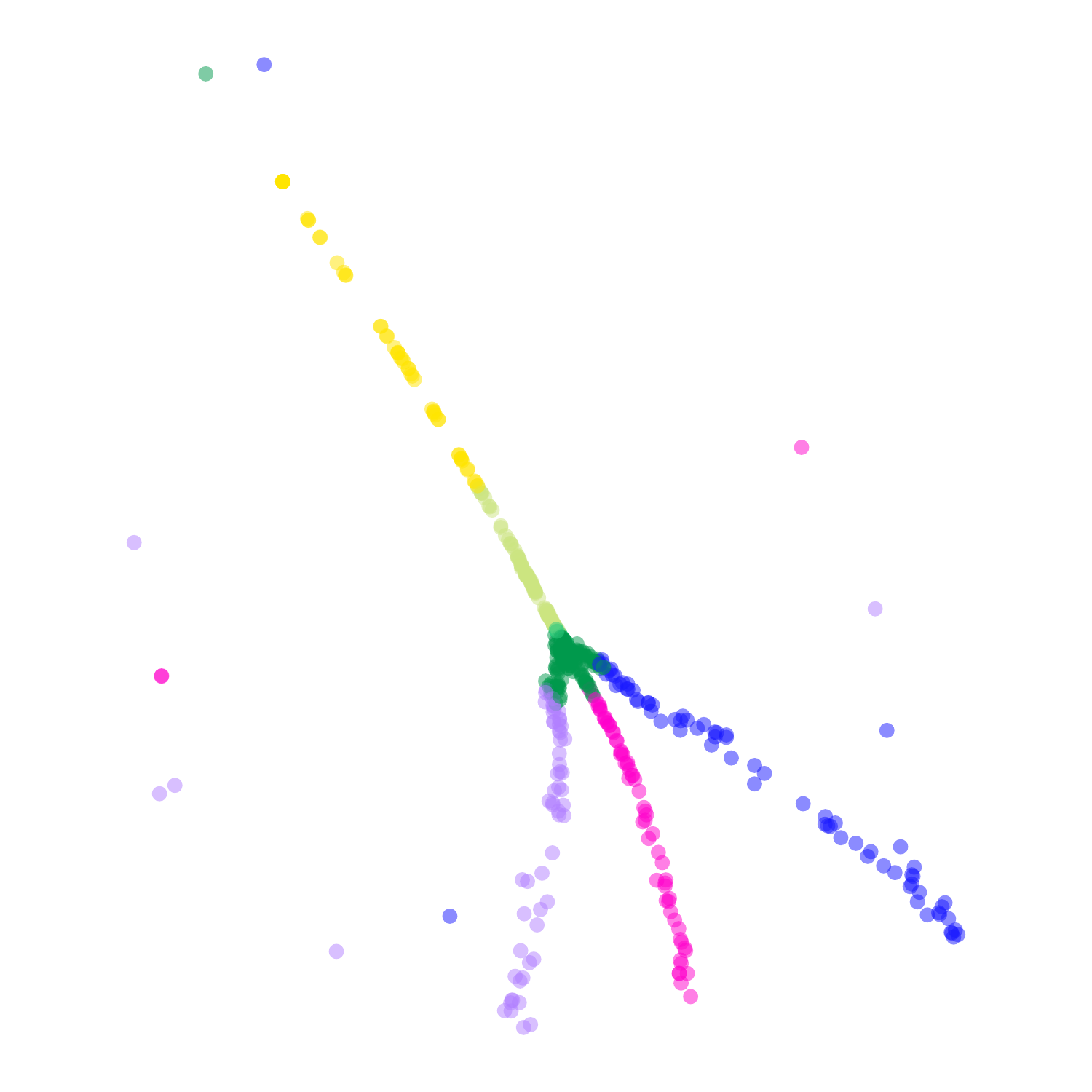}
    \end{subfigure} &
    \begin{subfigure}[!h]{.28\linewidth}
        \centering
        \includegraphics[width=\linewidth]{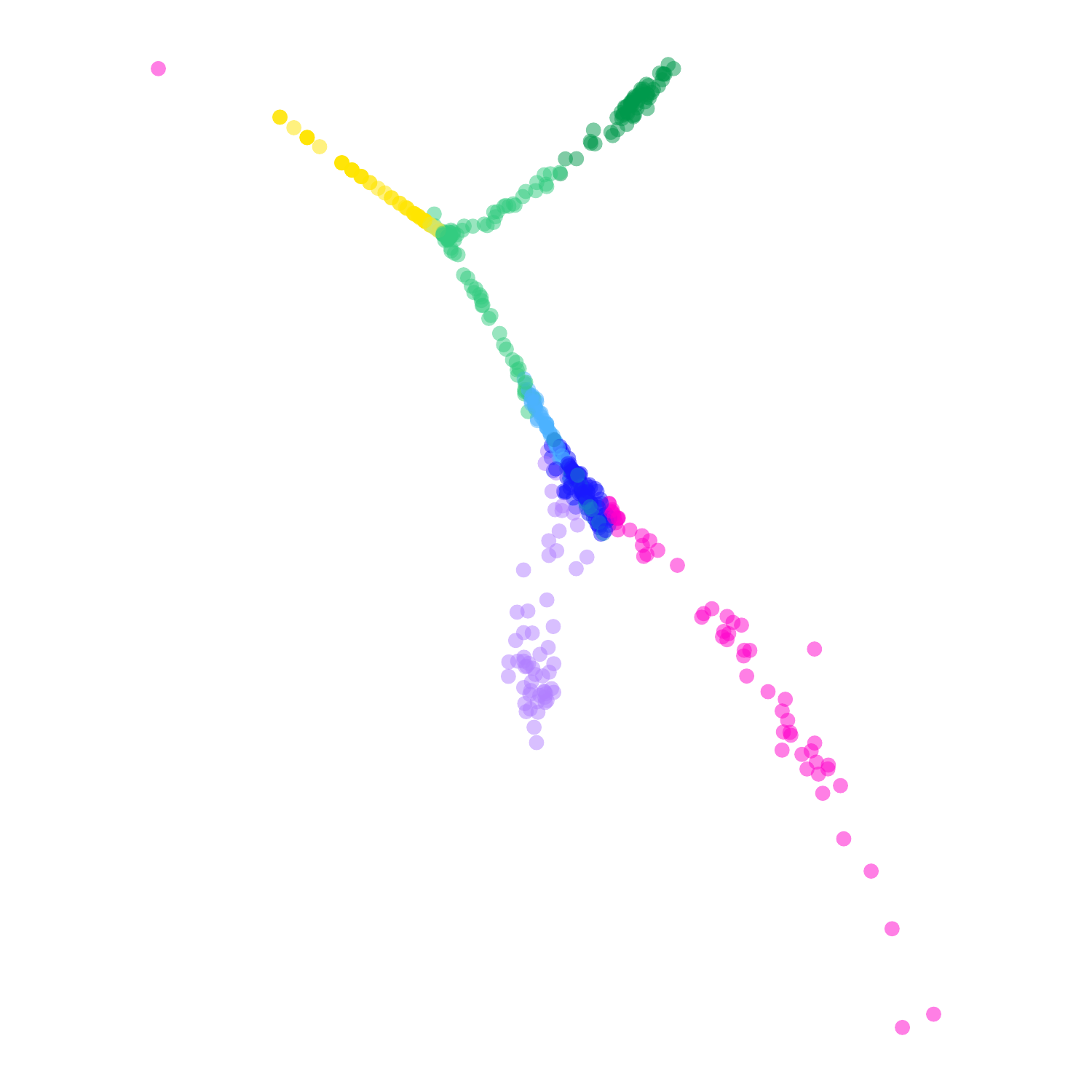}
    \end{subfigure} \\\\
    \rotatebox[origin=c]{90}{\small Step 4} &
    \begin{subfigure}[!h]{.28\linewidth}
        \centering
        \includegraphics[width=\linewidth]{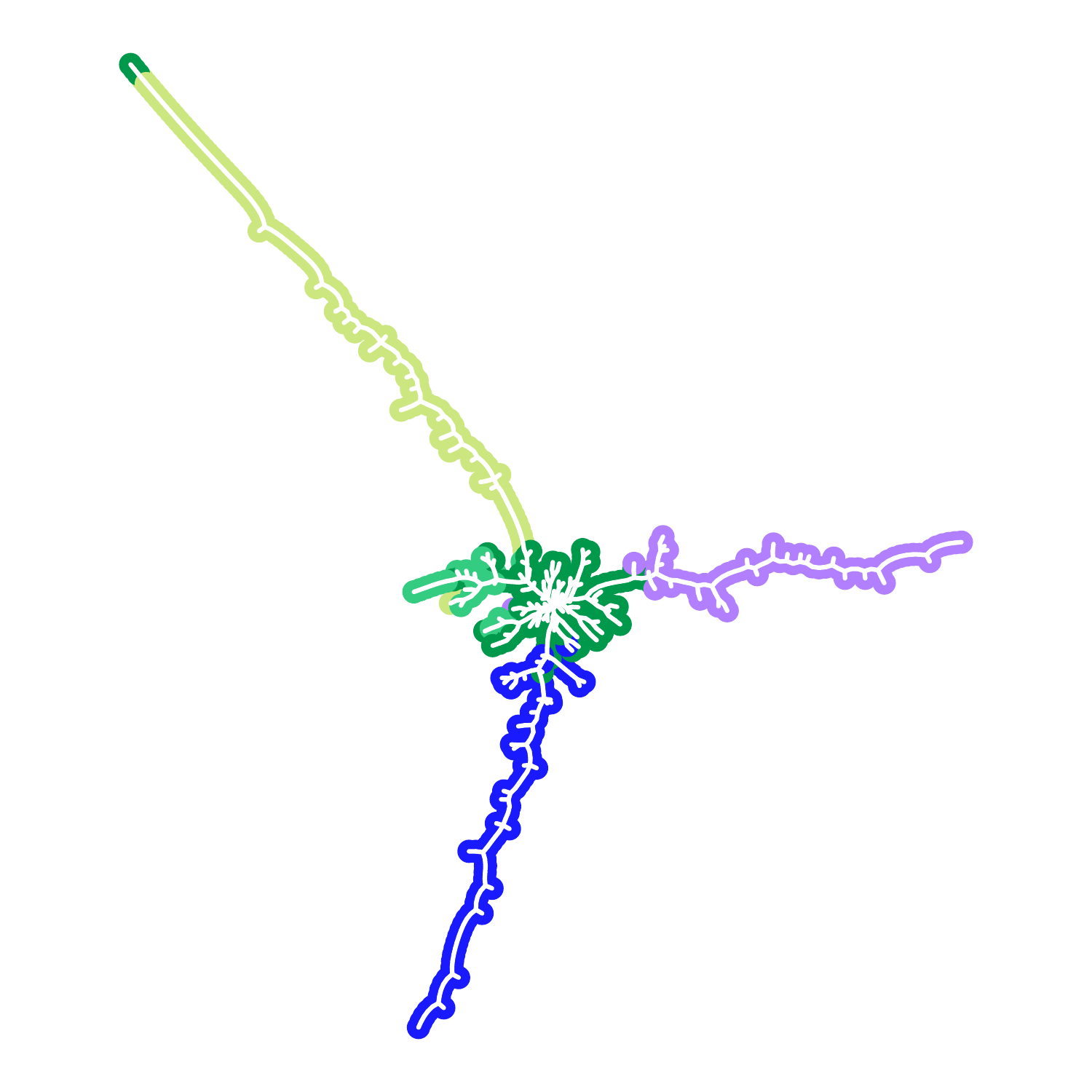}
    \end{subfigure} &
    \begin{subfigure}[!h]{.28\linewidth}
        \centering
        \includegraphics[width=\linewidth]{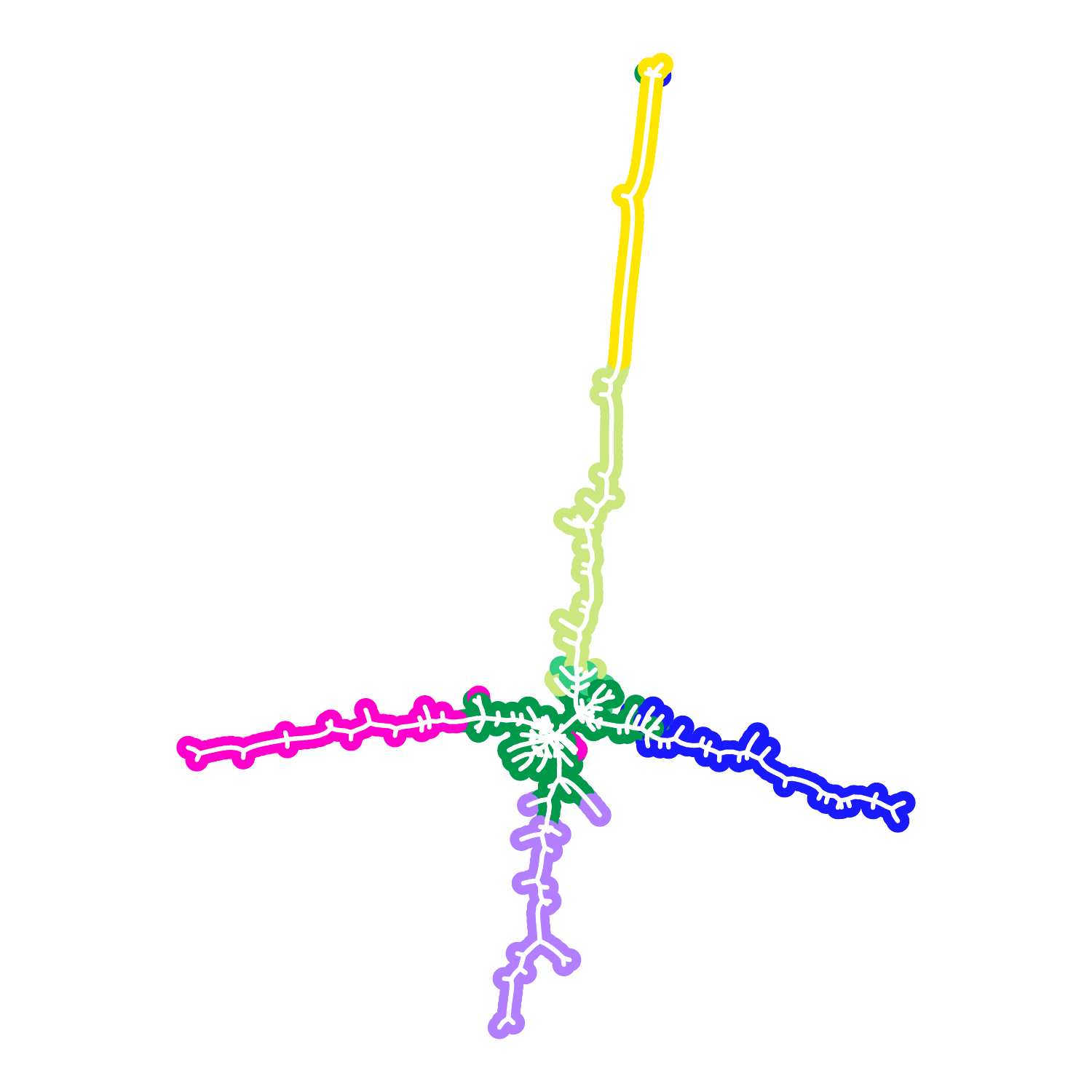}
    \end{subfigure} &
    \begin{subfigure}[!h]{.28\linewidth}
        \centering
        \includegraphics[width=\linewidth]{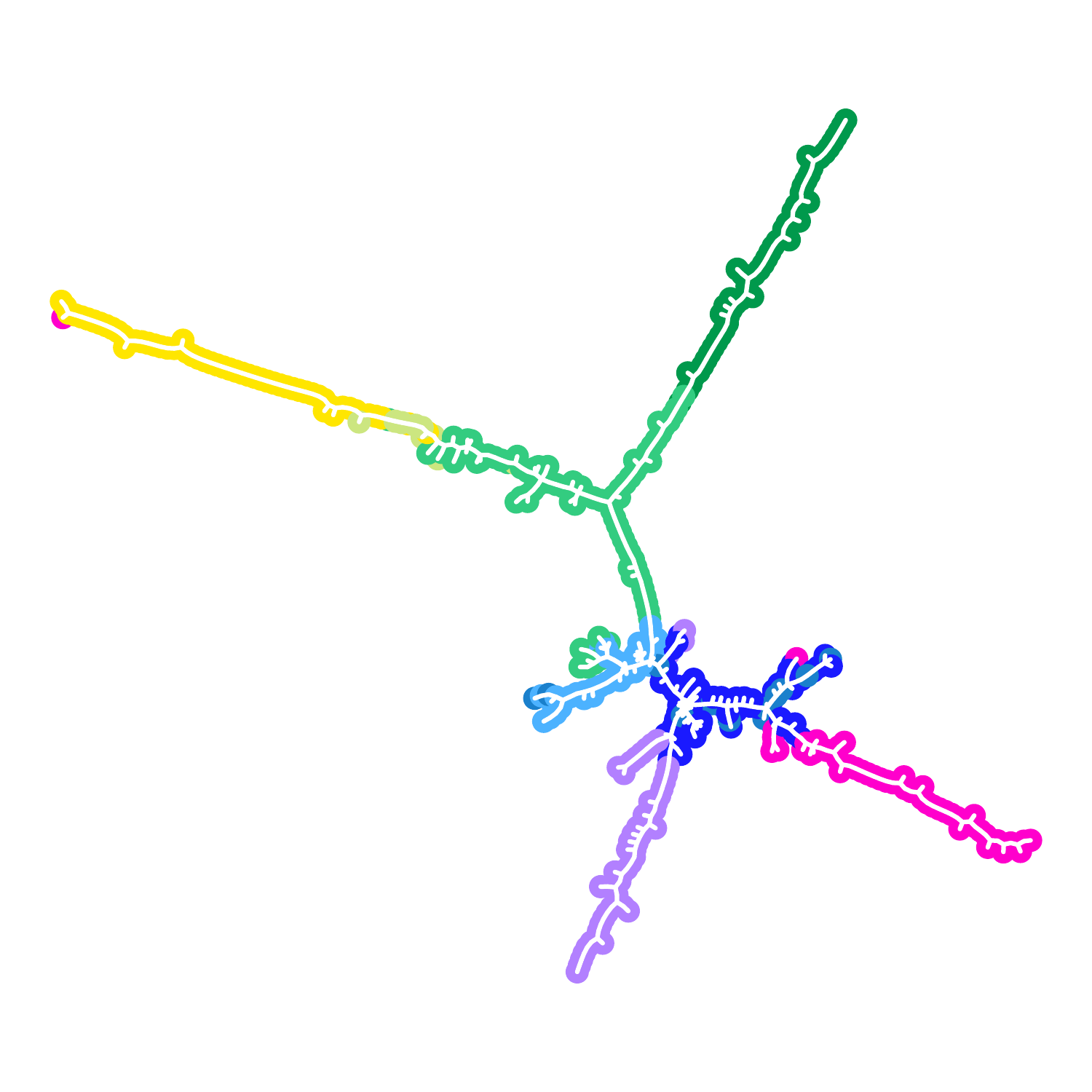} 
    \end{subfigure}
    \end{tabular}
    \vspace{.5cm}
    \caption{Dissimilarity matrix (Step 3) and corresponding tree (Step 4) inferred by our method for a bifurcating trajectory, a trifurcating trajectory, and a double bifurcating trajectory (from left to right) simulated using the \texttt{dyngen} library~\cite{robrecht_dyngen_2021}. Each dataset consists of $N=500$ cells and is originally in dimension $n=100$. We visualize dissimilarity matrices using multidimensional scaling (MDS) as implemented in the library \texttt{scikit-learn}~\cite{scikit-learn_2011}. All tree layouts are generated by the Kamada-Kawai algorithm available in the \texttt{NetworkX} Python library~\cite{networkx_2008}. In all three examples, both the dissimilarity matrix and the inferred tree accurately recover the underlying structure as well as the color-coded cell states. Occasional outliers may appear when integral curves fail to converge, in which case the corresponding nodes are attached to the root of the tree upon tree inference. Minor secondary branches can also emerge near branch nodes, as a result of the local cyclic dynamics inherent in the processes simulated by \texttt{dyngen} -- even in purely differentiating trajectories.}\label{fig:simulated}
    \vspace{-.6cm}
\end{figure}

\medbreak
Given the tree inferred by our method for each dataset, we can assign a pseudotime to any cell in the dataset, defined as the depth of the corresponding node in the tree, scaled to a value between $0$ (root) and $1$ (deepest leaf). We then compare this pseudotime to the pseudotime automatically assigned to the cells in \texttt{dyngen} during simulation, as well as to those produced by VeTra~\cite{weng_vetra_2021} and CellPath~\cite{zhang_cellpath_2021}, two other trajectory inference methods that allow for inferring differentiation trees from RNA velocity fields. Unlike our method, however, VeTra and CellPath do not return a single pseudotime for the entire dataset, but only one for each detected branch of the tree. To compare these two methods with our method, we first aggregate these branch-wise pseudotimes into a single one by simply assigning to each cell the pseudotime(s) assigned to it within a branch (provided that the cell itself has been assigned to at least one branch). We observe that our method is by far the one that yields the most accurate pseudotime (Fig.~\ref{fig:comparative}). 

\medbreak
That said, the pseudotime inferred by our method and the ground-truth pseudotime differ in their parameterization. Indeed, the pseudotime generated by \texttt{dyngen} for a dataset corresponds to the arc length of the simulated trajectory as described by the generated expression profiles in $\bR^d$, whereas the pseudotime assigned by our method to each cell is proportional to the number of preceding cells. Hence, our method does not seek to infer any metric information about the trajectory, but only its topology. To quantitatively assess the performance of our method with respect to the ground-truth, we therefore propose comparing pseudotemporal orderings rather than pseudotimes themselves. To do this, given a trajectory, we compute for each of its branches -- as defined by the backbone -- the percentage of pairs of cells that are correctly ordered by our method, and measure the accuracy of our method as the average percentage over all branches. Similarly, we measure the accuracy of VeTra and CellPath, relying on the previously aggregated pseudotime. To avoid the bias introduced by such aggregation, another way of measuring the accuracy of VeTra (resp. CellPath) would be to first match each ground-truth branch with one of the branches detected by the method, and then compare pseudo-temporal orderings only for the cells that belong to the detected branch. However, this option actually yields poorer performance for both methods. Across the three datasets, our method achieves significantly higher accuracy than VeTra and CellPath (Table~\ref{tab:comparative}). This can be explained by the fact that, in the case of VeTra and CellPath, a significant number of cells are not assigned to any branch (Fig.~\ref{fig:comparative}) and are therefore automatically wrongly ordered relative to all other cells. It should also be noted that CellPath achieves better overall results than VeTra, which is likely due to VeTra relying on a 2D embedding of the dataset (PCA in our implementation). 
\begin{figure}[p]
    \centering
    \begin{tabular}{cccc}
    \rotatebox[origin=c]{90}{\small Ground truth} &
    \begin{subfigure}[!h]{.28\linewidth}
        \centering
        \includegraphics[width=\linewidth]{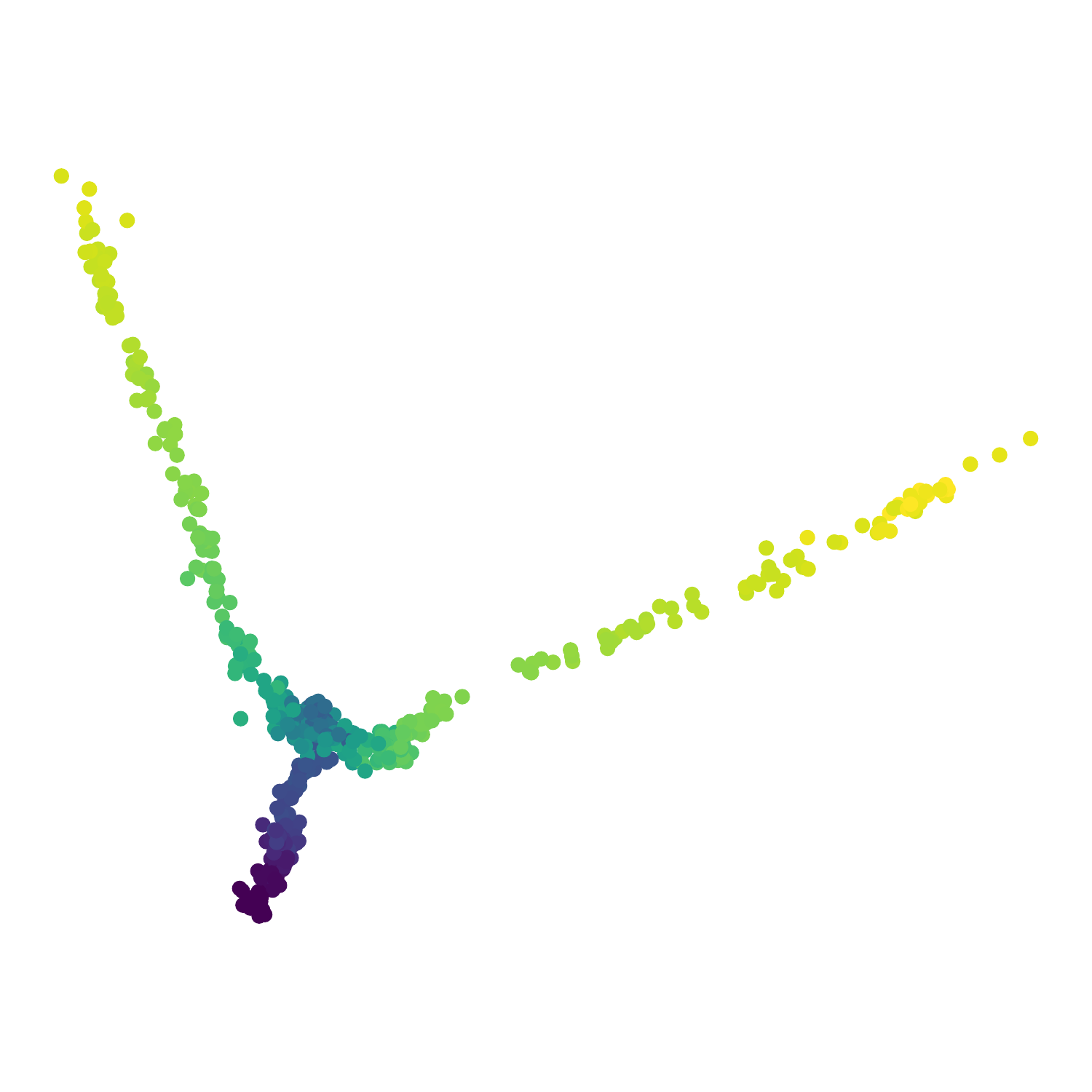}
    \end{subfigure} &
    \begin{subfigure}[!h]{.28\linewidth}
        \centering
        \includegraphics[width=\linewidth]{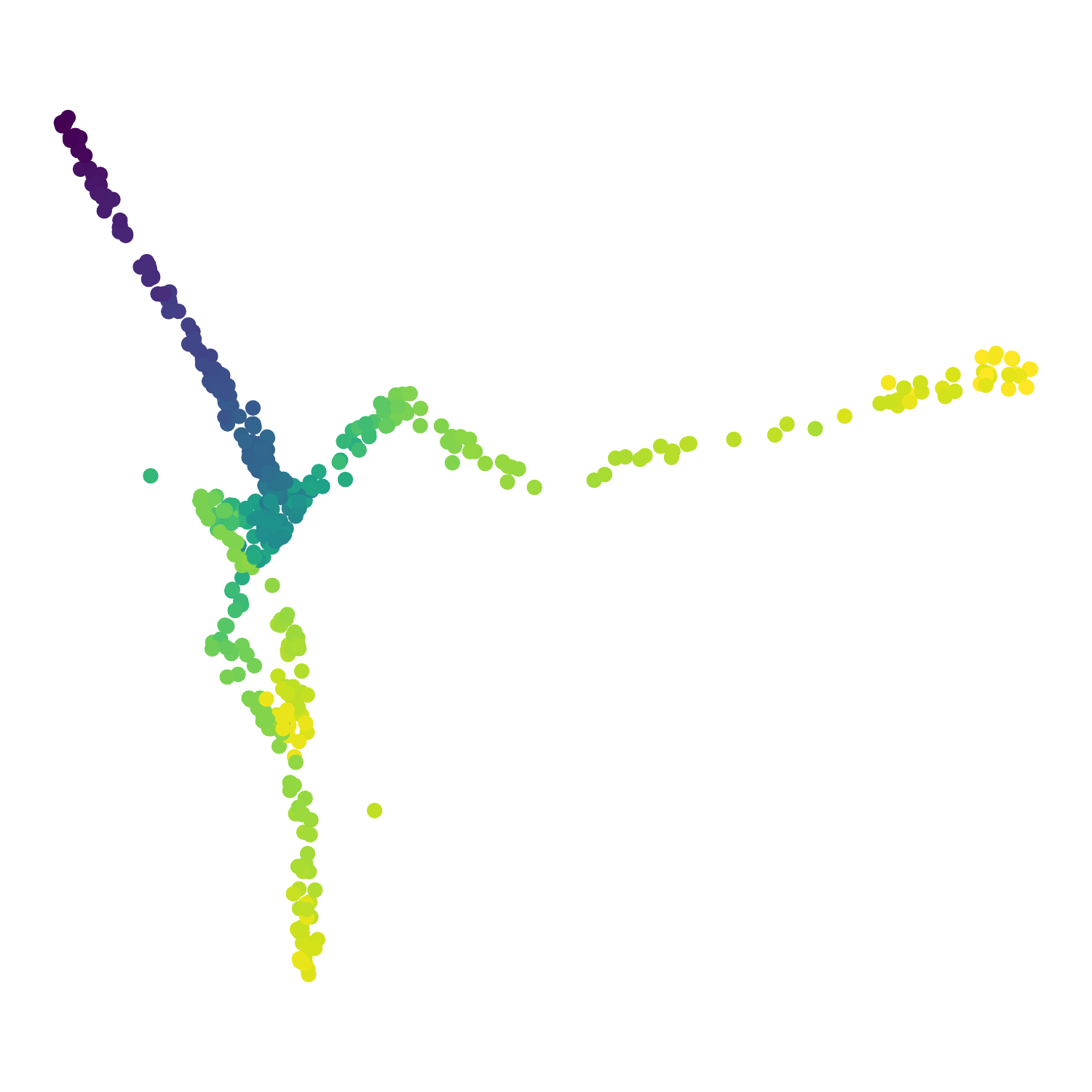}
    \end{subfigure} &
    \begin{subfigure}[!h]{.28\linewidth}
        \centering
        \includegraphics[width=\linewidth]{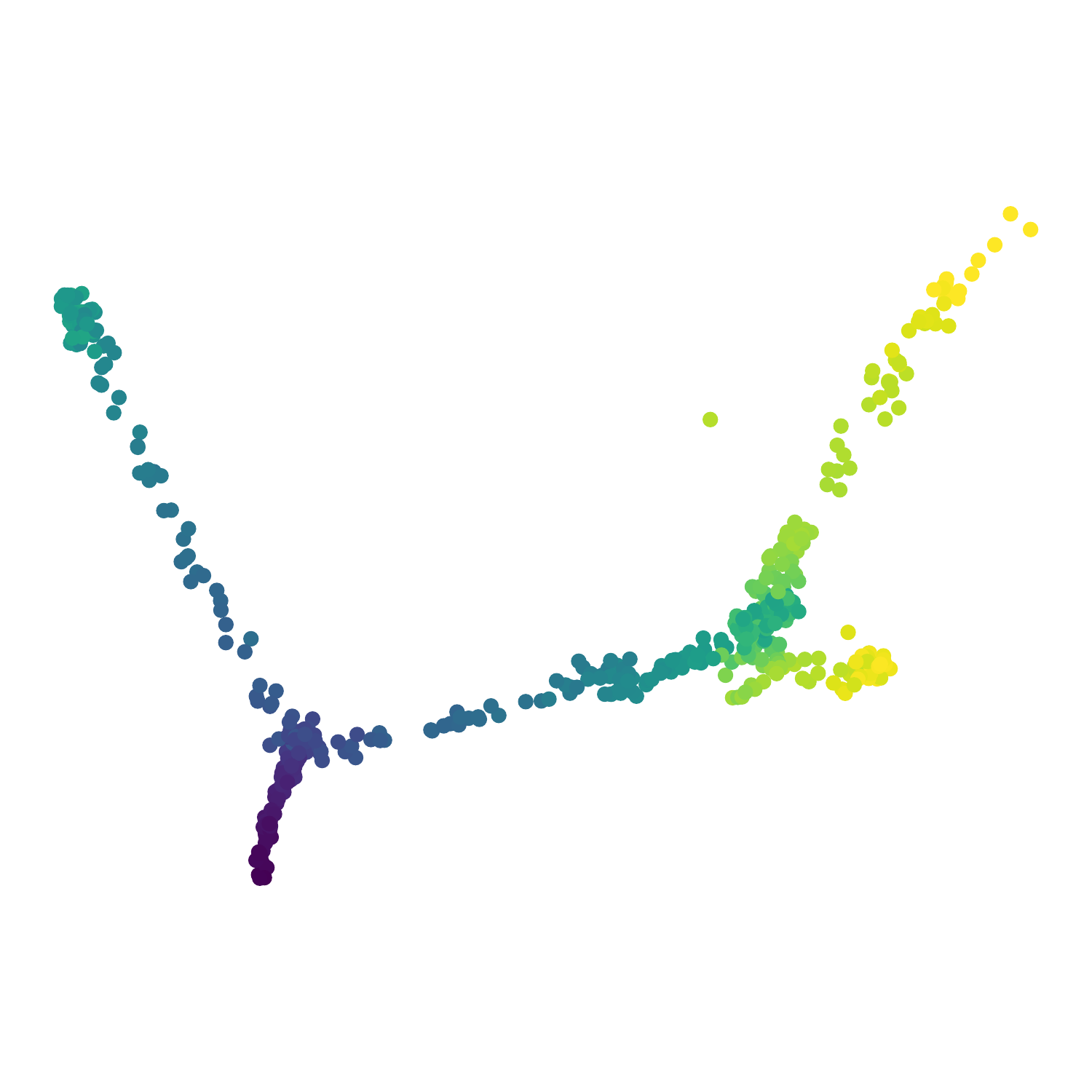}
    \end{subfigure} \\
    \rotatebox[origin=c]{90}{\small Ours} &
    \begin{subfigure}[!h]{.28\linewidth}
        \centering
        \includegraphics[width=\linewidth]{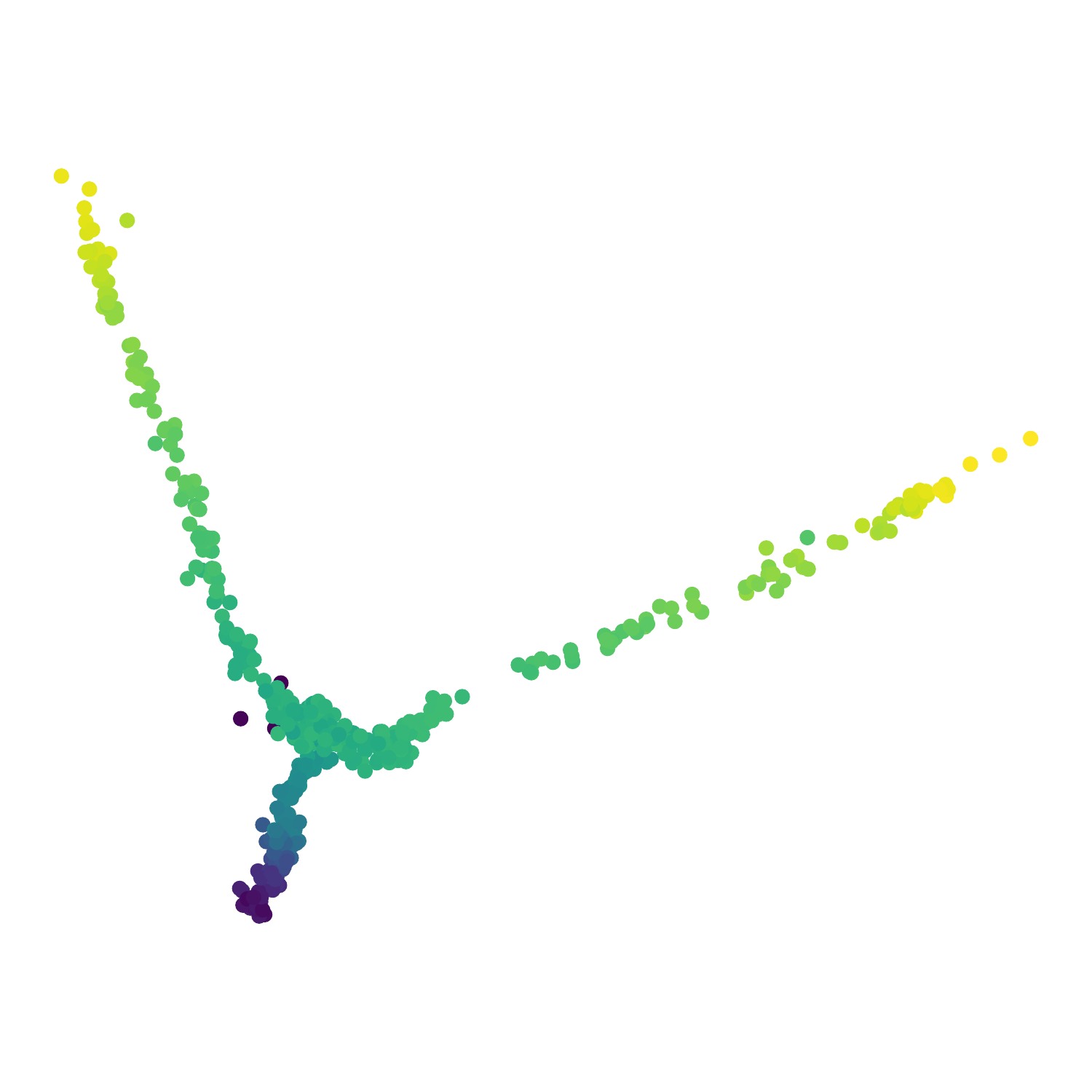}
    \end{subfigure} &
    \begin{subfigure}[!h]{.28\linewidth}
        \centering
        \includegraphics[width=\linewidth]{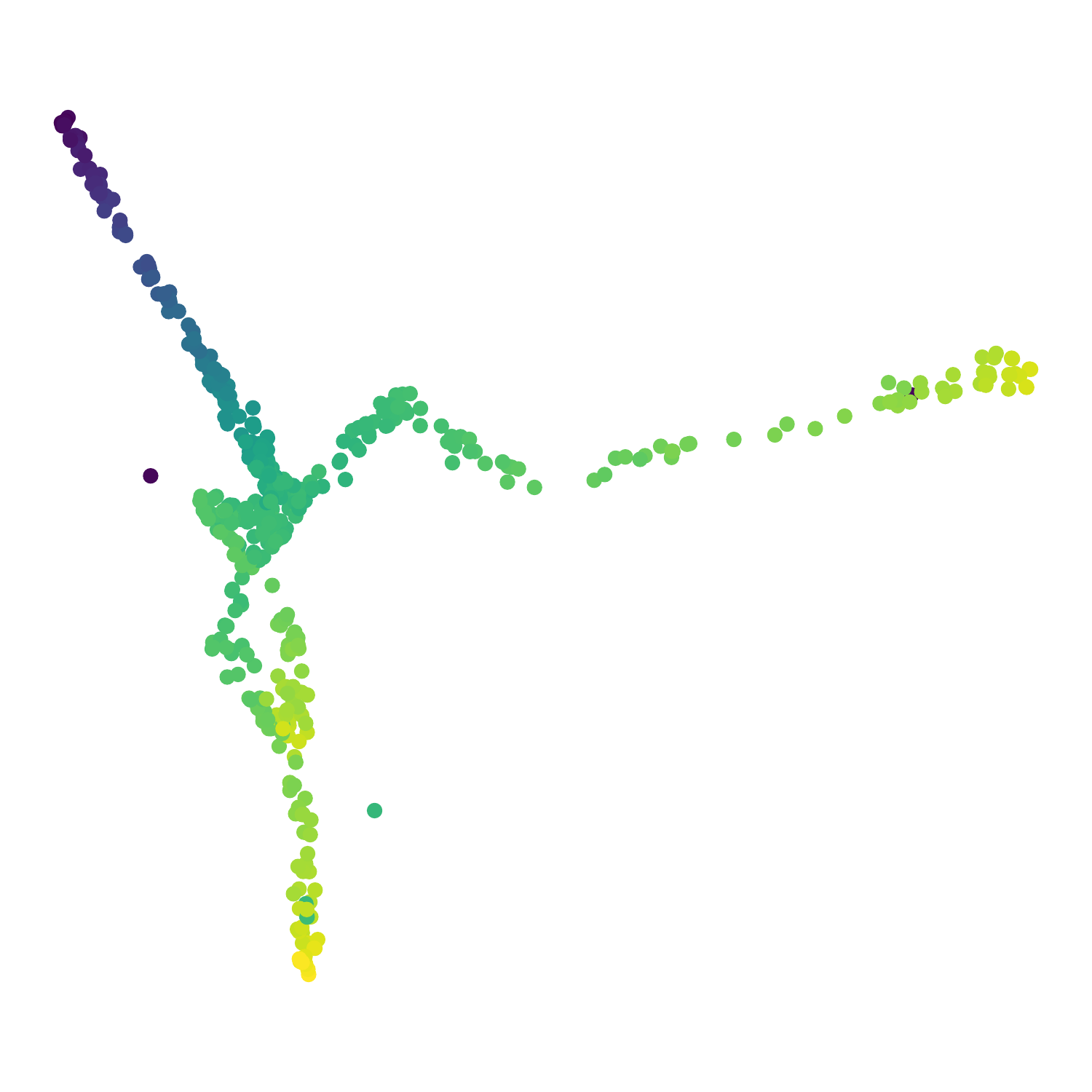}
    \end{subfigure} &
    \begin{subfigure}[!h]{.28\linewidth}
        \centering
        \includegraphics[width=\linewidth]{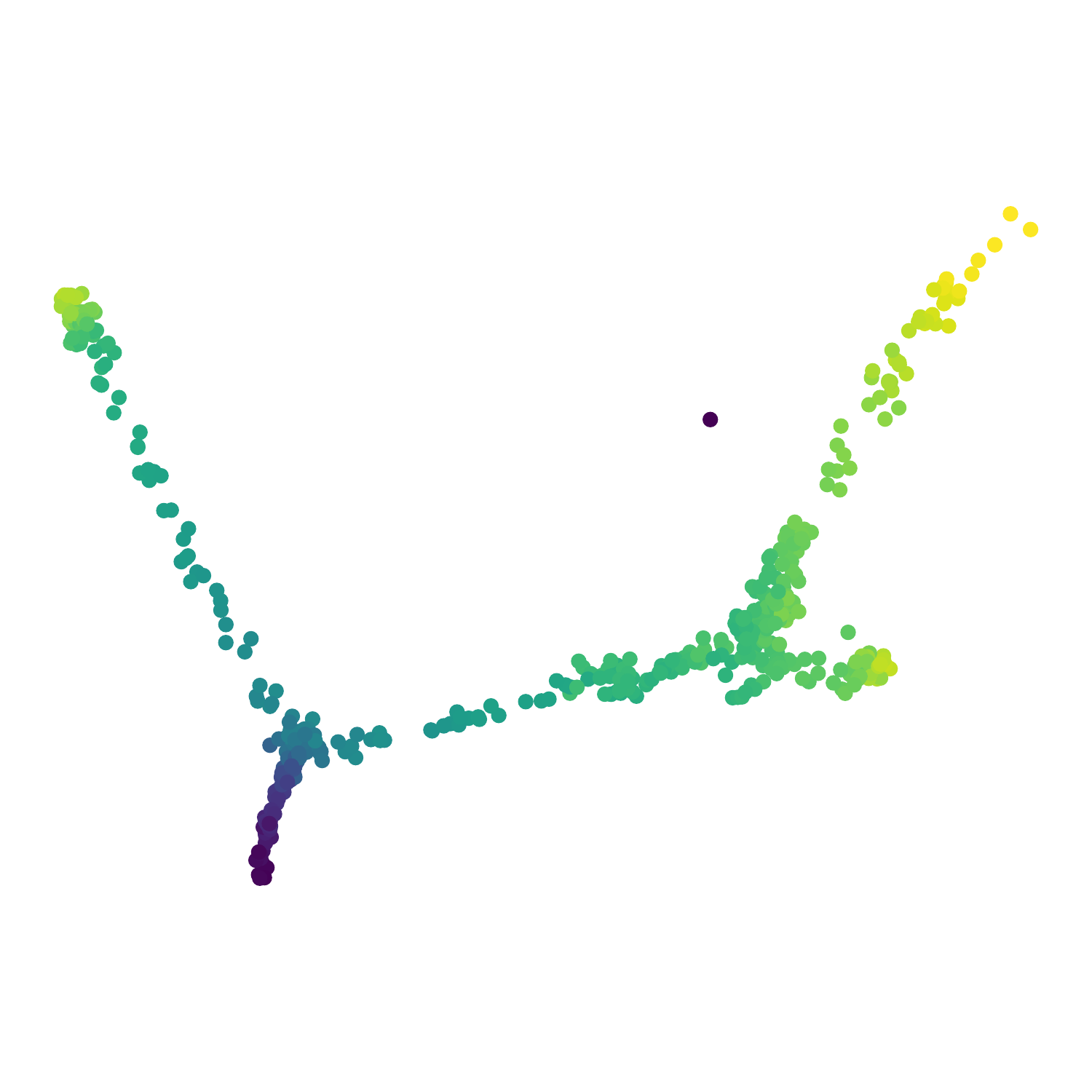} 
    \end{subfigure} \\
    \rotatebox[origin=c]{90}{\small VeTra} &
    \begin{subfigure}[!h]{.28\linewidth}
        \centering
        \includegraphics[width=\linewidth]{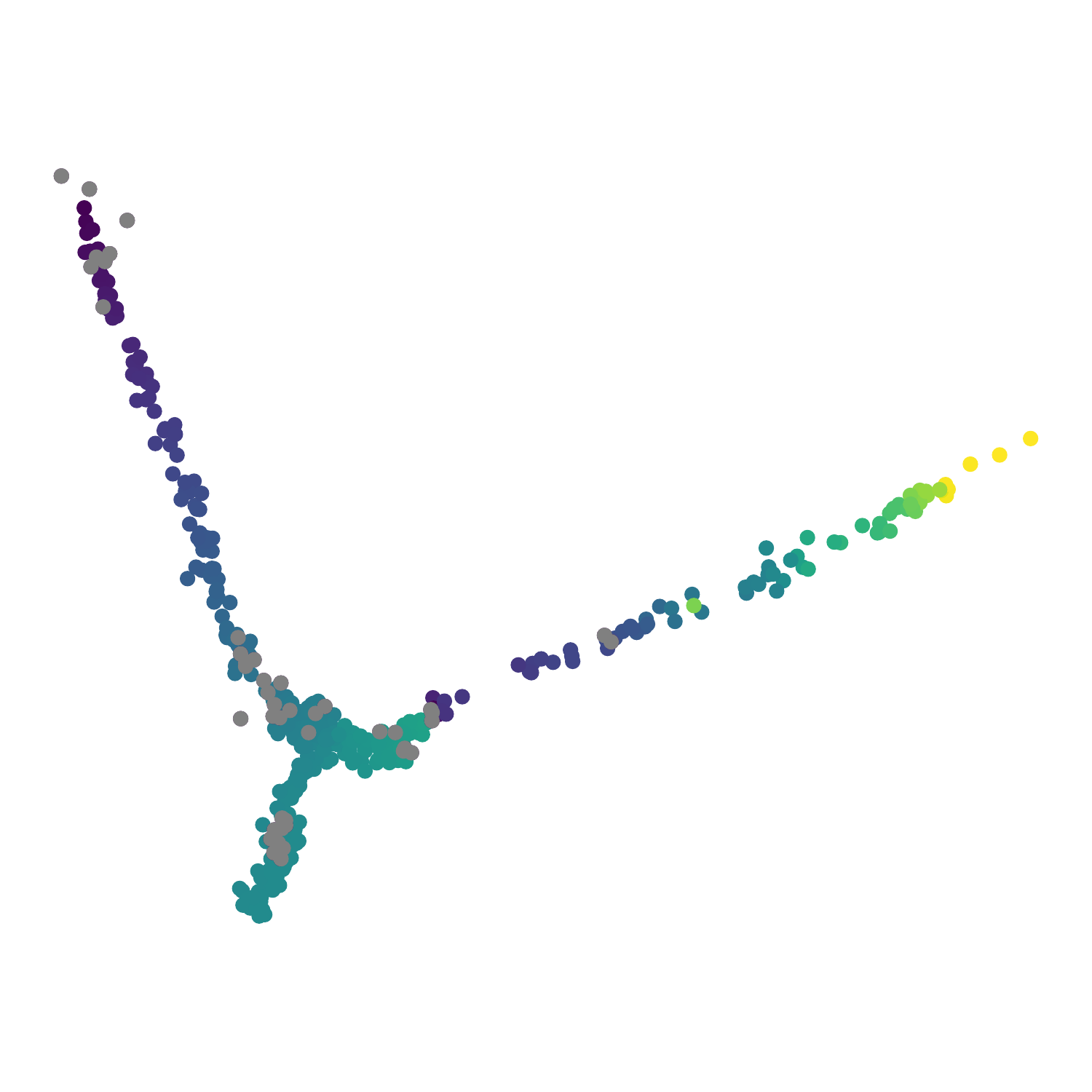}
    \end{subfigure} &
    \begin{subfigure}[!h]{.28\linewidth}
        \centering
        \includegraphics[width=\linewidth]{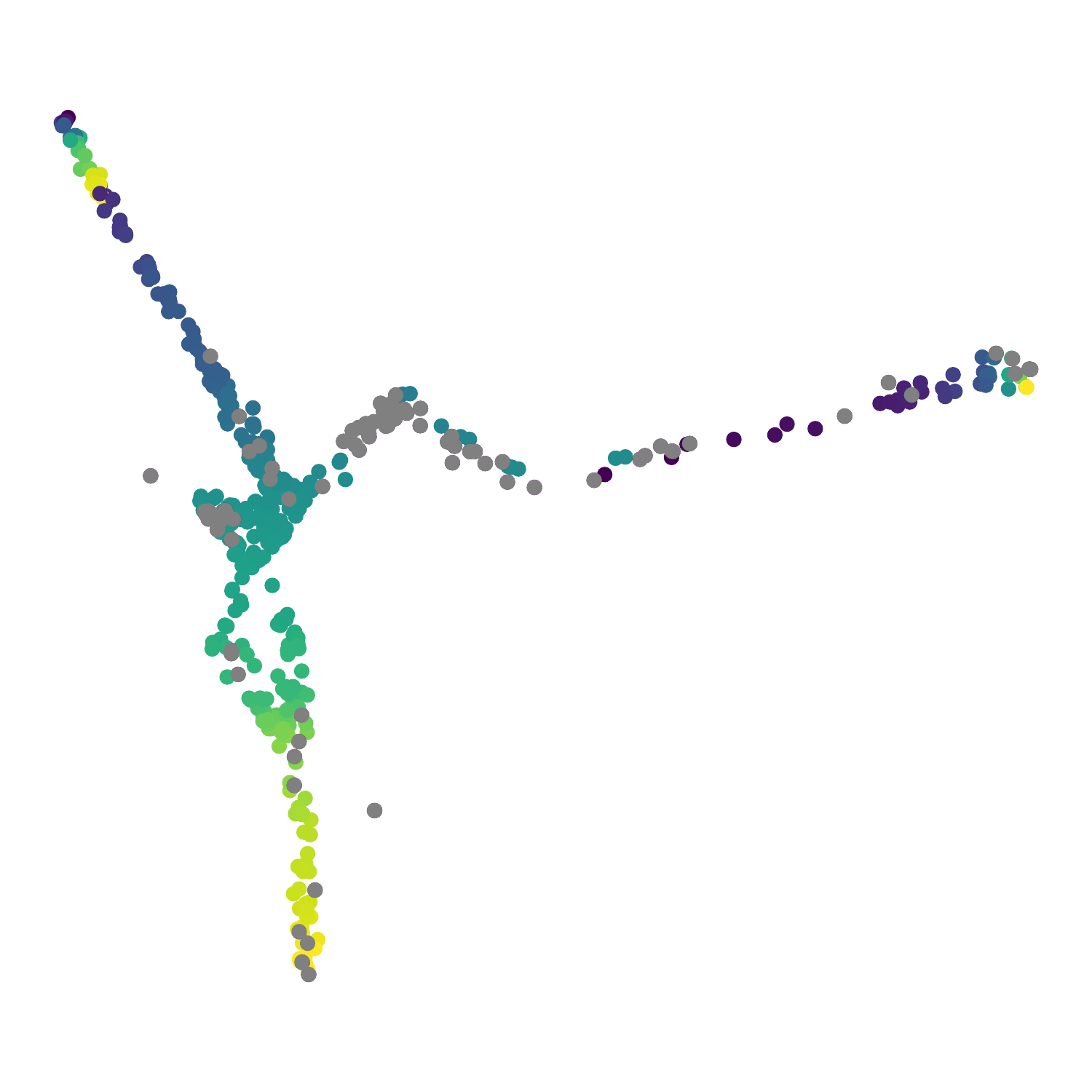}
    \end{subfigure} &
    \begin{subfigure}[!h]{.28\linewidth}
        \centering
        \includegraphics[width=\linewidth]{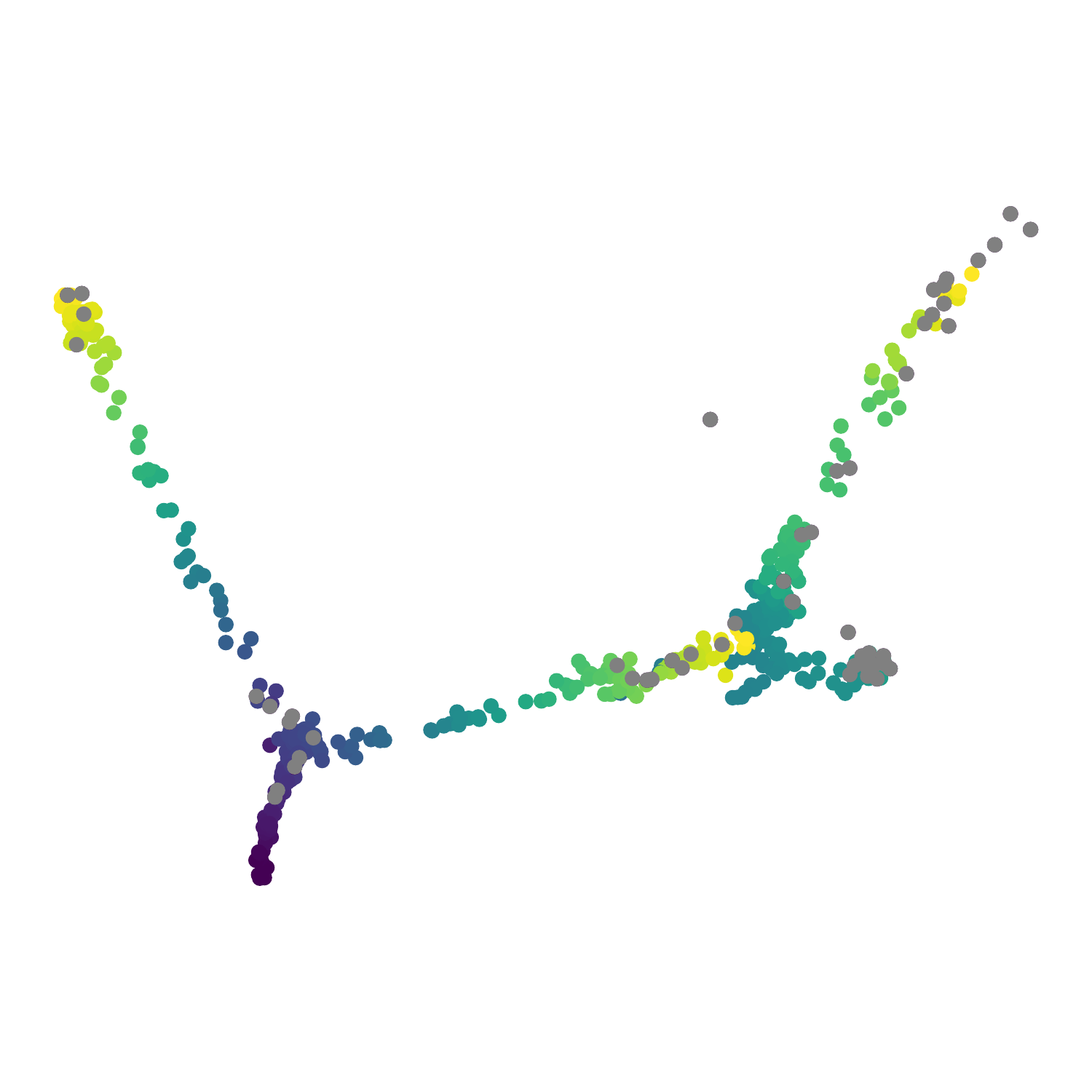} 
    \end{subfigure} \\
    \rotatebox[origin=c]{90}{\small CellPath} &
    \begin{subfigure}[!h]{.28\linewidth}
        \centering
        \includegraphics[width=\linewidth]{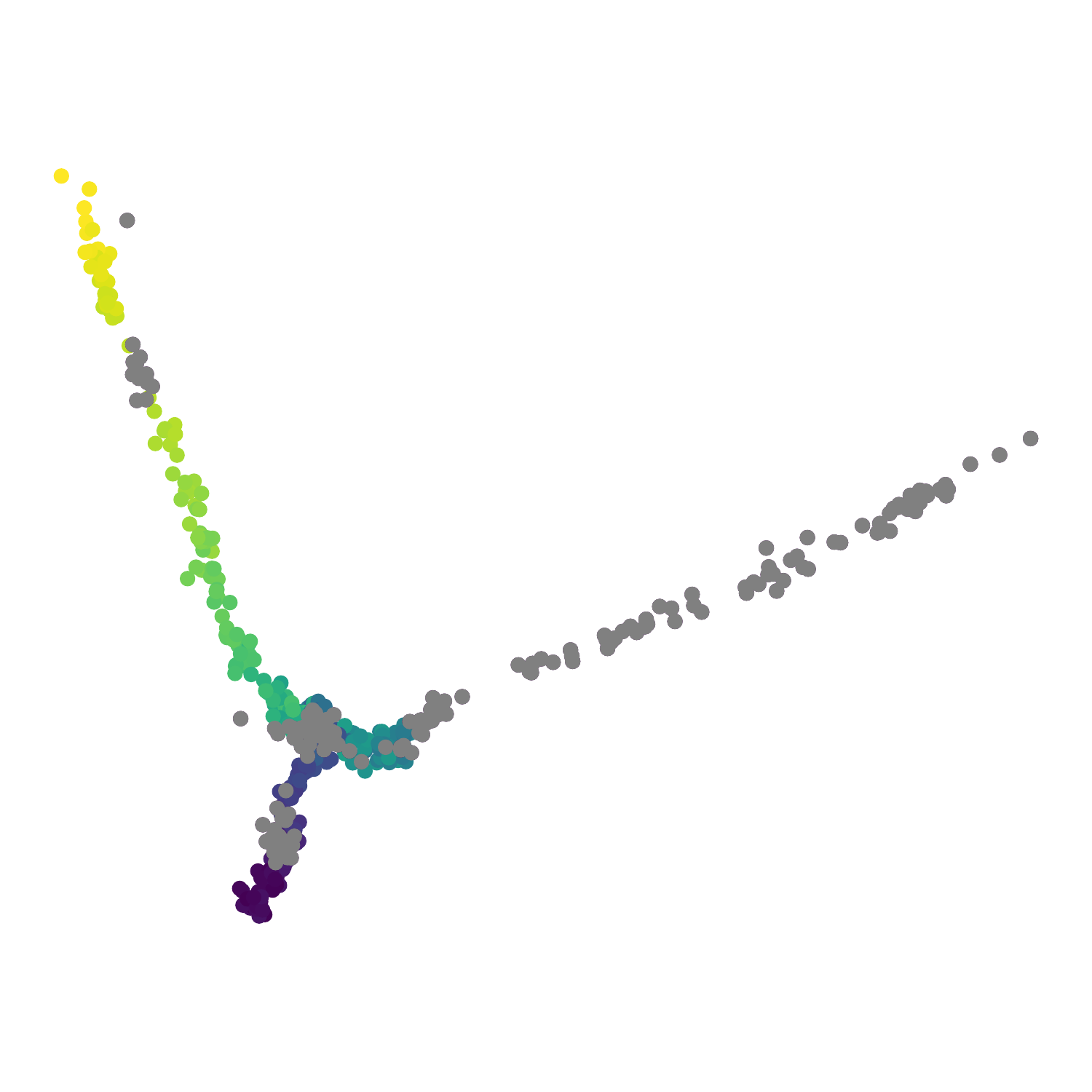}
    \end{subfigure} &
    \begin{subfigure}[!h]{.28\linewidth}
        \centering
        \includegraphics[width=\linewidth]{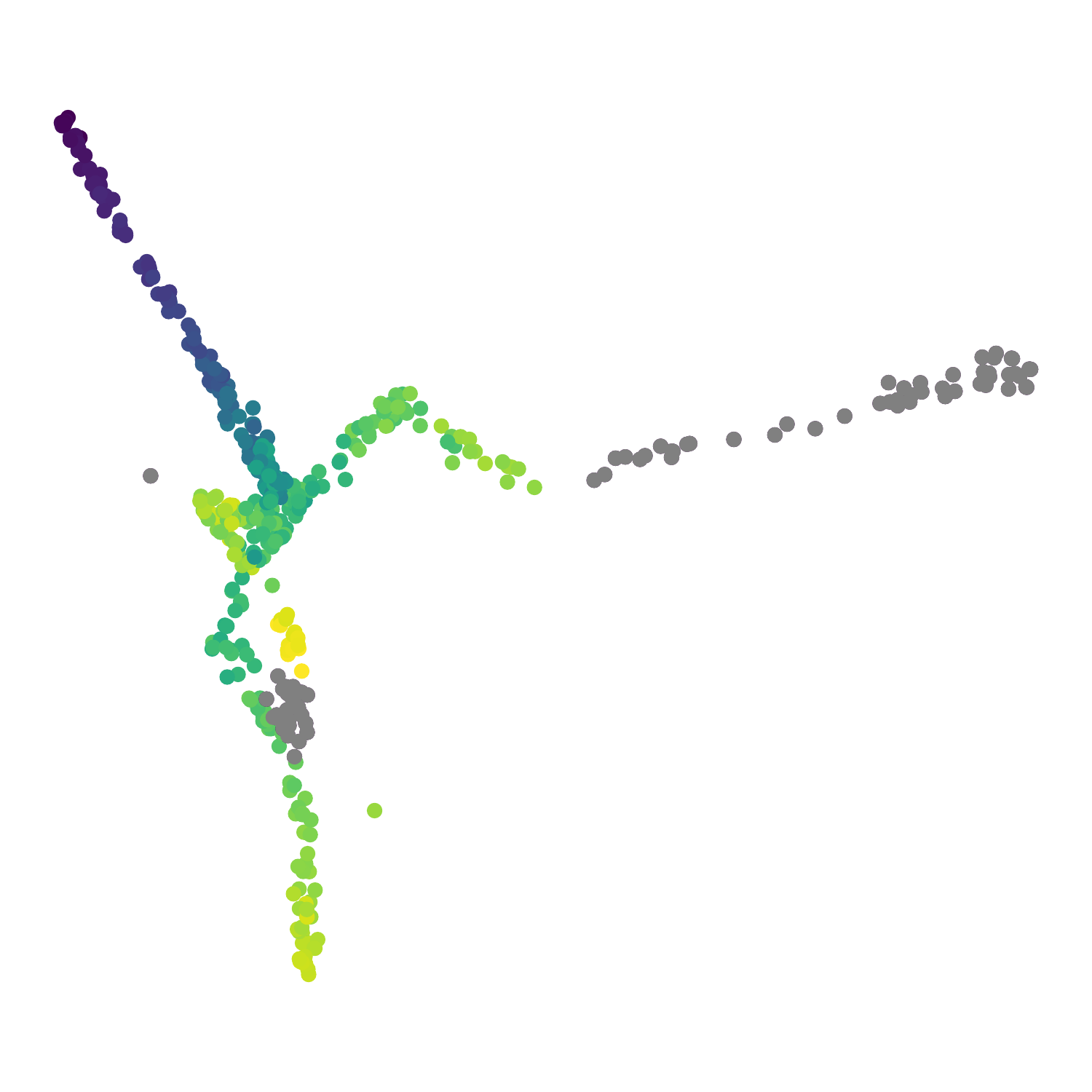}
    \end{subfigure} &
    \begin{subfigure}[!h]{.28\linewidth}
        \centering
        \includegraphics[width=\linewidth]{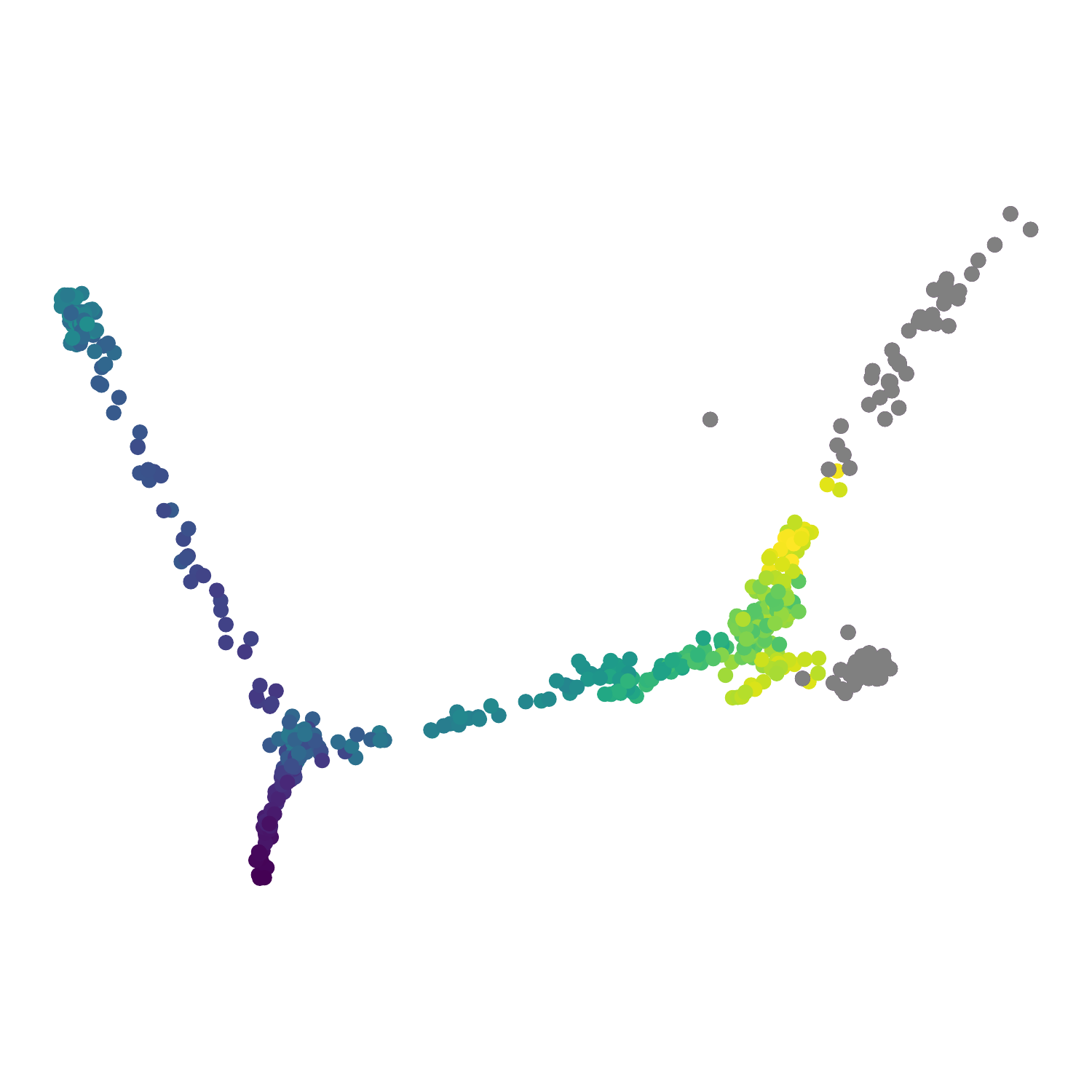} 
    \end{subfigure}
    \end{tabular}
    \vspace{.4cm}
    \caption{Pseudotime inferred by our method, Vetra, and CellPath for a bifurcating trajectory, a trifurcating trajectory, and a double-bifurcating trajectory (from left to right) simulated using the \texttt{dyngen} library~\cite{robrecht_dyngen_2021}. The 2D visualizations are generated via PCA. Shades of dark blue correspond to earlier cell states, while shades of yellow correspond to later cell states. Modulo some reparameterization —  from cell density to arc-length, the pseudotime inferred by our method agrees very closely with the ground truth. In contrast, VeTra and CellPath fail to assign a pseudotime to a significant proportion of cells (in gray). VeTra appears moreover to be unable to correctly identify early-stage cells.}\label{fig:comparative}
\end{figure}

\begin{table}[!h]
    \centering
    \begin{tabular}{|cccc|}
        \hline
        Method & Bifurcating & Trifurcating & Double bifurcating \\
        \hline
        Ours & 88.6\% & 93.1\% & 92.4\% \\
        Vetra & 36.4\% & 57.2\% & 70.1\% \\
        CellPath & 44.8\% & 78.5\% & 74.4\% \\
        \hline
    \end{tabular}
    \caption{Average proportion of well-ordered cells.} \label{tab:comparative}
    \vspace{-1.3cm}
\end{table}

\subsection{Pancreatic endocrine cell dataset}
Finally, we apply our method to the RNA sequencing dataset of mouse pancreatic endocrine cells~\cite{bastidasponce_comprehensive_2019}, precisely to the corresponding built-in dataset integrated into the \texttt{scVelo} library~\cite{bergen_generalizing_2020}. This dataset corresponds to the differentiation of pancreatic progenitor cells into four types of mature endocrine cells: $\alpha$-cells, $\beta$-cells, $\delta$-cells, and $\varepsilon$-cells. We exclude certain cells -- ductal cells and endocrine progenitors with low Ngn3 levels -- from the original dataset beforehand, as we know they undergo purely cyclic dynamics and therefore do not fit within the framework of tree inference. In fact, even the remaining cells do not follow a strict differentiation process, as it has been shown that $\alpha$-cells can originate from two different types of cell: pre-endocrine cells, but also $\varepsilon$-cells~\cite{yu_sequential_2021}. Our method can nevertheless be employed in this scenario: in practice, this mean that $\alpha$-cells are eventually distributed across two branches of the inferred tree (Fig.~\ref{fig:pancreas}). Although this distribution is not topologically accurate,  the inferred pseudotemporal order -- combined with cell annotations -- ultimately provides a successful prediction for the dual origin of $\alpha$-cells. On pure differentiation patterns (i.e., of pancreatic progenitor cells into $\beta$-cells), our method demonstrates, as expected, the most robust prediction. Note that setting the parameters of our method for this dataset (Table~\ref{tab:parameters}) proves to be significantly more difficult than for the simulated datasets due to the higher level of noise. In particular, choosing the appropriate scale $r$ to estimate the local dimension of the dataset during the preprocessing step (Section~\ref{sec:projection}) is particularly challenging, since the dataset as a whole is itself close to being $1$-dimensional, which, combined with the high level of noise, makes it difficult to clearly detect the transition from the local scale to the global scale. 
\begin{figure}[p]
    \centering
    \begin{tabular}{cc}
    \rotatebox[origin=c]{90}{\small Step 3} &
    \begin{subfigure}[!h]{.7\linewidth}
        \centering
        \includegraphics[width=\linewidth]{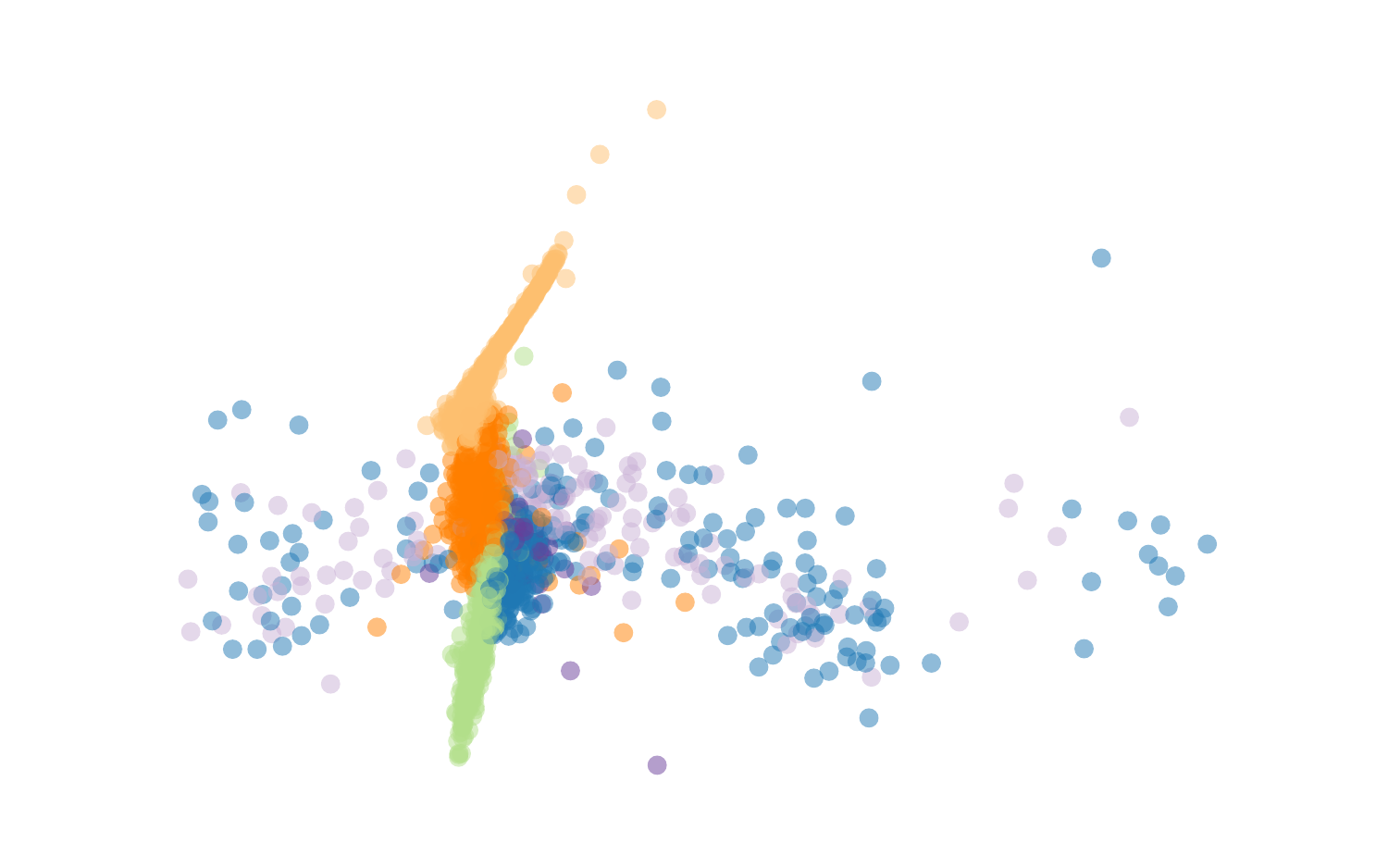}
    \end{subfigure} \\
    \rotatebox[origin=c]{90}{\small Step 4} &
    \begin{subfigure}[!h]{.7\linewidth}
        \centering
        \includegraphics[width=\linewidth]{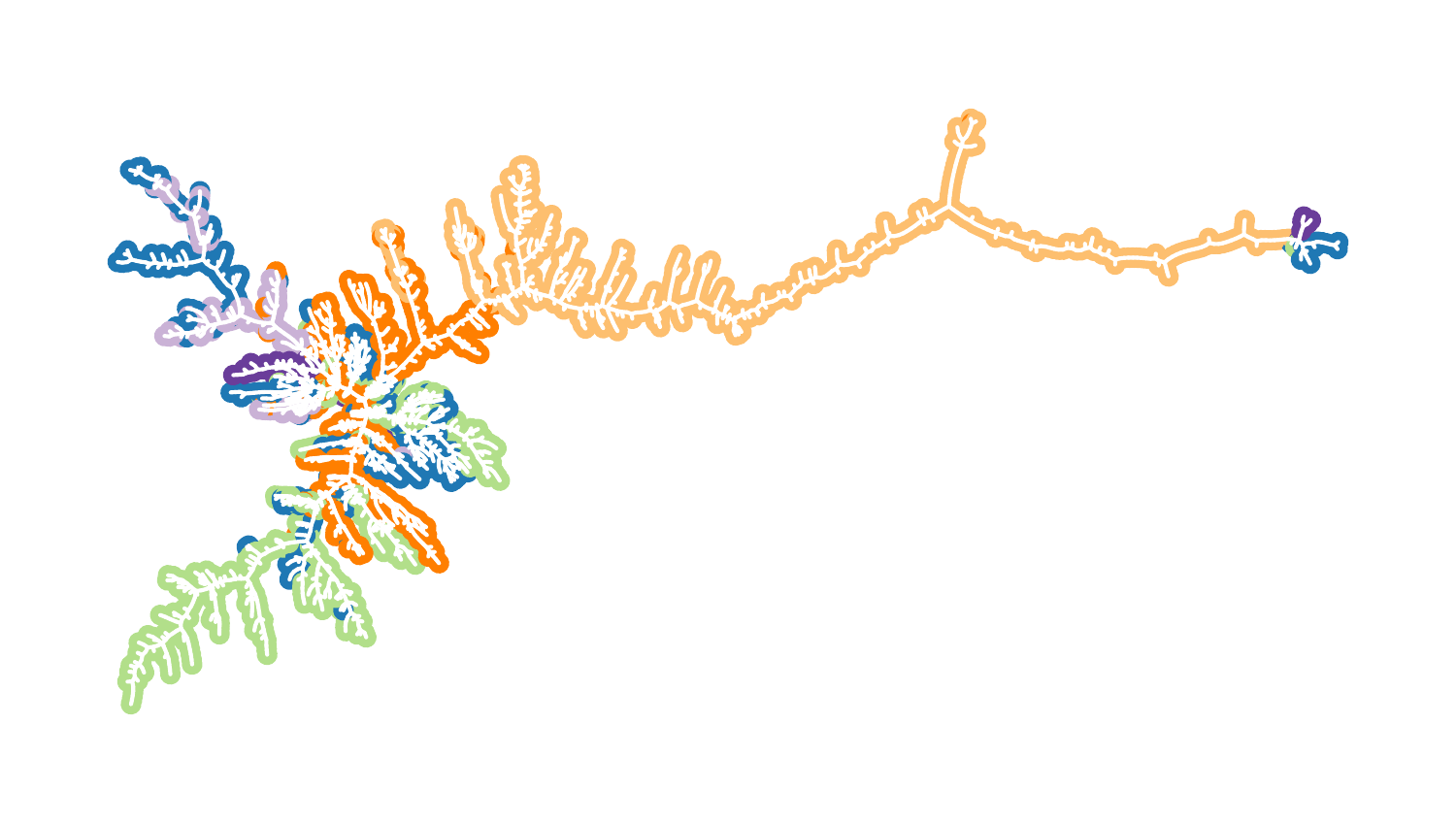}
    \end{subfigure}
    \end{tabular}
    \caption{Dissimilarity matrix (Step 3) and corresponding tree (Step 4) inferred by our method for the RNA sequencing dataset of mouse pancreatic endocrine cells~\cite{bastidasponce_comprehensive_2019}. The dissimilarity matrix is visualized using multidimensional scaling (MDS), and the tree layout is generated by the Kamada-Kawai algorithm. For improved readability, outliers have been removed from the visualization of the dissimilarity matrix. These outliers are still visible in the inferred tree, upstream of the root node. We identify two major branches in the tree. The lower branch accurately reconstructs the development of pre-endocrine cells (in dark orange) into $\beta$-cells (in green). Meanwhile, the upper branch corresponds to the development of $\varepsilon$-cells (in lilac) into $\alpha$-cells (in blue). A third, less major branch stems from the same node and consists of $\alpha$-cells originating directly from pre-endocrine cells. This branch is, however, more difficult to identify than the two first branches, as it is less elongated and also includes $\beta$-cells. In fact, younger $\alpha$-cells are challenging to separate from $\beta$-cells. In the visualization of the dissimilarity matrix, the path from $\varepsilon$-cells to $\alpha$-cells is split itself into two clusters: this is primarily an artifact of MDS itself. With the exception of outliers, $\delta$-cells (in purple) are well clustered in the inferred tree. However, because there are relatively few of them in the dataset, the branch they form is difficult to interpret as a meaningful branch, as it is no longer than minor branches attributable solely to noise.}\label{fig:pancreas}
\end{figure}

\begin{table}[!h]
    \begin{tabular}{|ccccc|}
        \hline
        Parameter & Bifurcating & Trifurcating & Double bifurcating & Pancreas\\
        \hline
        $d$ & 15 & 15 & 10 & 15 \\
        $t_{\text{diff}}$ & $2^{4.6}$ & $2^{8.2}$ & $2^{6.7}$ & $2^{5.2}$ \\
        $r$ & 4.5 & 7.5 & 8.5 & 11.0 \\
        $k$ & 40 & 40 & 40 & 90 \\
        $\sigma_x$ & 66.0 & 180.0 & 100.0 & 200.0 \\
        $\sigma_t$ & 1.21 & 1.16 & 1.06 & 0.85 \\
        \hline
    \end{tabular}
    \caption{Parameter settings for the different datasets.} \label{tab:parameters}
    \vspace{-1.3cm}
\end{table}



\section{Conclusion}\label{sec:conclusion}
In this article, we present a novel method for inferring differentiation trees from single-cell RNA sequencing (scRNA-seq) data that include both gene expression profiles and RNA velocity. Our method implements a distance-based tree inference approach, for which we introduce a new dissimilarity measure on scRNA-seq data, defined as the squared varifold distance between the integral curves of the RNA velocity field. We validate our method both with theoretical guarantees on the accuracy of the dissimilarity matrix to characterize the target tree and through a thorough evaluation of the method on simulated and real datasets. In particular, we show that our method significantly improves upon the state of the art on simulated benchmark datasets. In its current iteration, our method only addresses inferring the topology of the target differentiation tree, implying that all cells are uniformly distributed across pseudo-time. A natural extension of this work would therefore be to enable the inference of weighted differentiation trees. As part of a longer-term development of our method, we also want to tackle topologies other than trees, particularly cyclic topologies. To do this, we must adapt the distance-based inference algorithm. Incidentally, this should also allow us to handle a wider range of real-world datasets. 



\backmatter








\section*{Declarations}


\textbf{Funding}
This research is funded by the Deutsche Forschungsgemeinschaft (DFG, German Research Foundation) under Germany's Excellence Strategy – The Berlin Mathematics Research Center MATH+ (EXC-2046/1, project ID: 390685689) and part of this work was carried out at the Erwin Schrödinger International Institute for Mathematics and Physics (ESI) during the Thematic Programme on Infinite-dimensional Geometry: Theory and Applications. The authors have no conflict of interest to declare.

\medbreak
\noindent\textbf{Ethics approval and consent to participate}
Not applicable.

\medbreak
\noindent\textbf{Consent for publication}
All authors agreed with the content of this article and gave explicit consent to submit.

\medbreak
\noindent\textbf{Data availability}
The RNA sequencing dataset of mouse pancreatic endocrine cells is publicly available and can be downloaded from within the \texttt{scVelo} library at \url{https://github.com/theislab/scvelo_notebooks/tree/master/data/Pancreas}. The datasets simulated with the \texttt{dygen} library are available at \url{https://github.com/elodiemaignant/VeloTree}.

\medbreak
\noindent\textbf{Materials availability}
Not applicable.

\medbreak
\noindent\textbf{Code availability} The code is available at \url{https://github.com/elodiemaignant/VeloTree}.

\medbreak
\noindent\textbf{Author contribution}
Conceptualization: E.~Maignant, C.~von~Tycowicz;
Data curation: T.~Conrad, E.~Maignant, C.~von~Tycowicz;
Formal analysis: E.~Maignant, C.~von~Tycowicz, T.~Conrad;
Funding acquisition: C.~von~Tycowicz, T.~Conrad;
Investigation: E.~Maignant, C.~von~Tycowicz, T.~Conrad;
Methodology: E.~Maignant, C.~von~Tycowicz;
Project administration: C.~von~Tycowicz, T.~Conrad;
Resources: T.~Conrad;
Software: E.~Maignant, C.~von~Tycowicz;
Supervision: C.~von~Tycowicz, T.~Conrad;
Validation: E.~Maignant, C.~von~Tycowicz;
Visualization: E.~Maignant, C.~von~Tycowicz;
Writing – original draft: E.~Maignant, C.~von~Tycowicz;
Writing – review \& editing: E.~Maignant, C.~von~Tycowicz, T.~Conrad;







\begin{appendices}


\section{Proofs}\label{app:proofs}

\begin{proof}[Proof of lemma \ref{lemma:1}]
First, notice that all the quantities involved in Equation \ref{eq:L1} are positive. Therefore, we only need to prove that the quotient of the left-hand side scalar product by the right-hand side sum is bounded from above by a function that converges to $0$ when $\sigma_x$ does. 
\medbreak 
Let us start with bounding the left-hand side from above. Let $\varepsilon > 0$ and let $x_{j, \varepsilon}^- \in \gamma_{x_i, x_j}$ and $x_{j, \varepsilon}^+ \in \gamma_{x_j, x_k}$ such that
\begin{equation*}
    \ell(\gamma_{x_{j, \varepsilon}^-, x_j})=\ell(\gamma_{x_j, x_{j, \varepsilon}^+}) = \varepsilon.
\end{equation*} 
Thanks to the additivity of varifolds, we can decompose the left-hand side scalar product as follows 
\begin{align*}
    \langle \mu_{\gamma_{x_i, x_j}}, \mu_{\gamma_{x_j, x_k}} \rangle_{W^\ast} 
    = &\: \langle \mu_{\gamma_{x_{j, \varepsilon}^-, x_j}} , \mu_{\gamma_{x_j, x_{j, \varepsilon}^+}} \rangle_{W^\ast} \\
    &+ \langle \mu_{\gamma_{x_{j, \varepsilon}^-, x_j}} , \mu_{\gamma_{x_{j, \varepsilon}^+, x_k}} \rangle_{W^\ast} \\
    &+ \langle \mu_{\gamma_{x_i, x_{j, \varepsilon}^-}}, \mu_{\gamma_{x_j, x_{j, \varepsilon}^+}} \rangle_{W^\ast} \\
    &+ \langle \mu_{\gamma_{x_i, x_{j, \varepsilon}^-}}, \mu_{\gamma_{x_{j, \varepsilon}^+, x_k}} \rangle_{W^\ast}.
\end{align*}
Leveraging Assumption \ref{A1}, we derive then the following inequality 
\begin{align*}
    \langle \mu_{\gamma_{x_i, x_j}}, \mu_{\gamma_{x_j, x_k}} \rangle_{W^\ast} 
    \leq &\: \varepsilon^2 \\
    &+ \ell(\gamma_{x_j, x_k}) \varepsilon e^{-\|x_j - x_{j, \varepsilon}^+\|^2 / \sigma_x^2} \\
    &+ \ell(\gamma_{x_i, x_j}) \varepsilon e^{-\|x_{j, \varepsilon}^- - x_j\|^2 / \sigma_x^2} \\
    &+ \ell(\gamma_{x_i, x_j}) \ell(\gamma_{x_j, x_k}) e^{-\|x_{j, \varepsilon}^- - x_{j, \varepsilon}^+\|^2 / \sigma_x^2}.
\end{align*}
Now, because the paths are assumed to be smooth and regular, we have the following result on the length of an infinitesimal arc
\begin{equation*}
    \forall x,y \in \gamma_{x_i, x_k}, \: \|x -y\| \underset{x \to y}{\sim} \ell(\gamma_{x, y}).
\end{equation*}
The previous upper bound is then of the form
\begin{align*}
    \langle \mu_{\gamma_{x_i, x_j}}, \mu_{\gamma_{x_j, x_k}} \rangle_{W^\ast} 
    &\leq \varepsilon^2 + C_1 \varepsilon e^{-\frac{\varepsilon^2 + o(\varepsilon^2)}{\sigma_x^2}} + C_2e^{-4\frac{\varepsilon^2 + o(\varepsilon^2)}{\sigma_x^2}}
\end{align*}
where $C_1$ and $C_2$ do not depend on either $\varepsilon$, $\sigma_x$ or $\sigma_t$. 
\medbreak  
Then, let us bound the two norms in the right-hand side from below. For any $0 < \rho < \min(\ell(\gamma_{x_i, x_j}), \ell(\gamma_{x_j, x_k}))$, we can write
\begin{equation*}
    \|\mu_{\gamma_{x_i, x_j}}\|_{W^\ast}^2 \geq \iint_{\gamma_{x_i, x_j}^2}\mathbbm{1}_{\{y\in B(x, \rho)\}}e^{-\frac{\|x - y\|^2}{\sigma_x^2}}e^{-\frac{\|\vec{t}(x) - \vec{t}(y)\|^2}{\sigma_t^2}}d\ell(x)d\ell(y)
\end{equation*}
The curves $\gamma_{x_i, x_j}$ and $\gamma_{x_j, x_k}$ are assumed to be orientable, that is $x \mapsto \vec{t}(x)$ is a smooth map on both domains. Therefore it is $K$-Lipschitz for some $K>0$ and we have
\begin{equation*}
    \|\mu_{\gamma_{x_i, x_j}}\|_{W^\ast}^2 \geq \iint_{\gamma_{x_i, x_j}^2}\mathbbm{1}_{\{y\in B(x, \rho)\}}e^{-\frac{\rho^2}{\sigma_x^2}}e^{-\frac{K^2\rho^2}{\sigma_t^2}}d\ell(x)d\ell(y)\geq C_3\rho e^{-\frac{\rho^2}{\sigma_x^2}}e^{-K^2\frac{\rho^2}{\sigma_t^2}} 
\end{equation*}
and similarly $\|\mu_{\gamma_{x_j, x_k}}\|_{W^\ast}^2 \geq C_4\rho e^{-\frac{\rho^2}{\sigma_x^2}}e^{-K^2\frac{\rho^2}{\sigma_t^2}}$ where $C_3$ and $C_4$ do not depend on either $\rho$, $\sigma_x$ or $\sigma_t$.
\medbreak
Finally, let us fix $\varepsilon=\sigma_x^p$ for some $\frac{1}{2} < p < 1$ and $\rho=\sigma_x$. Then we obtain an upper bound on the quotient that converges to $0$ when $\sigma_x$ converges to $0$ with the condition that $\sigma_x \asymp \sigma_t$. \qedsymbol
\end{proof}

\begin{proof}[Proof of lemma \ref{lemma:2}]
Again, notice that all the quantities involved are positive. Therefore, we only need to prove that the quotient of the left-hand side scalar product by the right-hand side sum is bounded from above by a function that converges to $0$ when $\sigma_t$ does. 
\medbreak    
Let us bound the left-hand side from above. Let $\varepsilon > 0$ such that $B(x_i, \varepsilon) \subset U_i$ and let $x_{i, \varepsilon}^L \in \gamma_{x_i, x_j}$ and $x_{i, \varepsilon}^R \in \gamma_{x_i, x_k}$ such that
\begin{equation*}
    \ell(\gamma_{x_i, x_{i, \varepsilon}^L})=\ell(\gamma_{x_i, x_{i, \varepsilon}^R}) = \varepsilon.
\end{equation*} 
We have that
\begin{align*}
    \langle \mu_{\gamma_{x_i, x_j}}, \mu_{\gamma_{x_j, x_k}} \rangle_{W^\ast} 
    = &\: \langle \mu_{\gamma_{x_i, x_{i, \varepsilon}^L}} , \mu_{\gamma_{x_i, x_{i, \varepsilon}^R}} \rangle_{W^\ast} \\
    &+ \langle \mu_{\gamma_{x_i, x_{i, \varepsilon}^L}} , \mu_{\gamma_{x_{i, \varepsilon}^R, x_k}} \rangle_{W^\ast} \\
    &+ \langle \mu_{\gamma_{x_{i, \varepsilon}^L, x_j}}, \mu_{\gamma_{x_i, x_{i, \varepsilon}^R}} \rangle_{W^\ast} \\
    &+ \langle \mu_{\gamma_{x_{i, \varepsilon}^L, x_j}}, \mu_{\gamma_{x_{j, \varepsilon}^r, x_k}} \rangle_{W^\ast}.
\end{align*}
Leveraging Assumption \ref{A2}, we derive the following inequality
\begin{align*}
    \langle \mu_{\gamma_{x_i, x_j}}, \mu_{\gamma_{x_j, x_k}} \rangle_{W^\ast} 
    \leq &\: \varepsilon^2 \\
    &+ \ell(\gamma_{x_i, x_k}) \varepsilon e^{-\|\vec{t}(x_i) - \vec{t}(x_{i, \varepsilon}^R)\|^2 / \sigma_t^2} \\
    &+ \ell(\gamma_{x_i, x_j}) \varepsilon e^{-\|\vec{t}(x_{i, \varepsilon}^L) - \vec{t}(x_i)\|^2 / \sigma_t^2} \\
    &+ \ell(\gamma_{x_i, x_j}) \ell(\gamma_{x_i, x_k}) e^{-\|\vec{t}(x_{i, \varepsilon}^L) - \vec{t}(x_{i, \varepsilon}^R)\|^2 / \sigma_t^2}.
\end{align*}
Leveraging then Assumption \ref{A3}, we obtain a lower bound of the form
\begin{align*}
    \langle \mu_{\gamma_{x_i, x_j}}, \mu_{\gamma_{x_j, x_k}} \rangle_{W^\ast} 
    &\leq \varepsilon^2 + C_1\varepsilon e^{-\frac{\varepsilon^{2a}}{\sigma_t^2}} + C_2e^{-4\frac{\varepsilon^{2a}}{\sigma_t^2}}
\end{align*}
where $C_1$ and $C_2$ do not depend on either $\varepsilon$, $\sigma_x$ or $\sigma_t$.
\medbreak
Then, let us bound the two norms in the right-hand side from below. The map $x \to \vec{t}(x)$ is $K$-Lipschitz on $\gamma_{x_i, x_j}$ and $\gamma_{x_i, x_k}$ for some $K>0$, and for any $0 < \rho < \min(\ell(\gamma_{x_i, x_j}), \ell(\gamma_{x_i, x_k}))$, we then have
\begin{equation*}
    \|\mu_{\gamma_{x_i, x_j}}\|_{W^\ast}^2 \geq \iint_{\gamma_{x_i, x_j}^2}\mathbbm{1}_{\{y\in B(x, \rho)\}}e^{-\frac{\rho^2}{\sigma_x^2}}e^{-\frac{K^2\rho^2}{\sigma_t^2}}d\ell(x)d\ell(y)\geq C_3\rho e^{-\frac{\rho^2}{\sigma_x^2}}e^{-K^2\frac{\rho^2}{\sigma_t^2}}
\end{equation*}
and similarly $\|\mu_{\gamma_{x_i, x_k}}\|_{W^\ast}^2 \geq C_4re^{-\frac{\rho^2}{\sigma_x^2}}e^{-K^2\frac{\rho^2}{\sigma_t^2}}$ where $C_3$ and $C_4$ do not depend on either $\rho$, $\sigma_x$ or $\sigma_t$.
\medbreak
Finally, let us fix $\varepsilon=\sigma_x^p$ for some $\frac{1}{2} < p < \frac{1}{a}$ and $\rho=\sigma_x$. Then we obtain an upper bound on the quotient that converges to $0$ when $\sigma_x$ converges to $0$ with the condition that $\sigma_x \asymp \sigma_t$. \qedsymbol
\end{proof}




\end{appendices}


\bibliography{main}

@inproceedings{kaltenmark_general_2017,
    title={A general framework for curve and surface comparison and registration with oriented varifolds},
    author={Kaltenmark, Irene and Charlier, Benjamin and Charon, Nicolas},
    booktitle={Proceedings of the IEEE conference on computer vision and pattern recognition},
    pages={3346--3355},
    year={2017}
}

@article{manno_rna_2018,
    author = {La Manno, G. and Soldatov, R. and Zeisel, A. and et al.},
    volume = {560},
    year = {2018},
    title = {RNA velocity of single cells},
    journal = {Nature},
    doi={10.1038/s41586-018-0414-6}
}

@article{gorin_rna_2022,
    title={RNA velocity unraveled},
    author={Gorin, Gennady and Fang, Meichen and Chari, Tara and Pachter, Lior},
    journal={PLOS Computational Biology},
    volume={18},
    number={9},
    pages={e1010492},
    year={2022},
    publisher={Public Library of Science San Francisco, CA USA},
    doi={10.1371/journal.pcbi.1010492}
}

@article{charlier_fshape_2017,
    title={The fshape framework for the variability analysis of functional shapes},
    author={Charlier, Benjamin and Charon, Nicolas and Trouv{\'e}, Alain},
    journal={Foundations of Computational Mathematics},
    volume={17},
    pages={287--357},
    year={2017},
    publisher={Springer},
    doi={10.1007/s10208-015-9288-2}
}

@article{pardi_combinatorics_2012,
    title={Combinatorics of distance-based tree inference},
    author={Pardi, Fabio and Gascuel, Olivier},
    journal={Proceedings of the National Academy of Sciences},
    volume={109},
    number={41},
    pages={16443--16448},
    year={2012},
    publisher={National Academy of Sciences},
    doi={10.1073/pnas.1118368109}
}

@article{zhang_cellpath_2021,
    title={Inference of high-resolution trajectories in single-cell RNA-seq data by using RNA velocity},
    author={Zhang, Ziqi and Zhang, Xiuwei},
    journal={Cell Reports Methods},
    volume={1},
    number={6},
    year={2021},
    publisher={Elsevier},
    doi={10.1016/j.crmeth.2021.100095}
}

@article{lange_cellrank_2022,
    title={CellRank for directed single-cell fate mapping},
    author={Lange, Marius and Bergen, Volker and Klein, Michal and Setty, Manu and Reuter, Bernhard and Bakhti, Mostafa and Lickert, Heiko and Ansari, Meshal and Schniering, Janine and Schiller, Herbert B and others},
    journal={Nature methods},
    volume={19},
    number={2},
    pages={159--170},
    year={2022},
    publisher={Nature Publishing Group US New York},
    doi={10.1038/s41592-021-01346-6}
}

@phdthesis{glaunes_these_2005,
    title={Transport par difféomorphismes de points, de mesures et de courants pour la comparaison de formes et l'anatomie numérique},
    author={Glaunès, Joan Alexis},
    year={2005}
}

@article{saelens_comparison_2019,
    title={A comparison of single-cell trajectory inference methods},
    author={Saelens, Wouter and Cannoodt, Robrecht and Todorov, Helena and Saeys, Yvan},
    journal={Nature biotechnology},
    volume={37},
    number={5},
    pages={547--554},
    year={2019},
    publisher={Nature Publishing Group US New York},
    doi={10.1038/s41587-019-0071-9}
}

@article{kalaghatgi_family_2016,
    author = {Kalaghatgi, Prabhav and Pfeifer, Nico and Lengauer, Thomas},
    title = {Family-Joining: A Fast Distance-Based Method for Constructing Generally Labeled Trees},
    journal = {Molecular Biology and Evolution},
    volume = {33},
    number = {10},
    pages = {2720-2734},
    year = {2016},
    doi = {10.1093/molbev/msw123}
}

@article{pereira_note_1969,
    title={A note on the tree realizability of a distance matrix},
    author={Pereira, JMS Simoes},
    journal={Journal of combinatorial theory},
    volume={6},
    number={3},
    pages={303--310},
    year={1969},
    publisher={Elsevier},
    doi={10.1016/S0021-9800(69)80092-X}
}

@Inbook{shan_diffusion_2022,
    author="Shan, Shan
    and Daubechies, Ingrid",
    editor="Flandrin, Patrick
    and Jaffard, St{\'e}phane
    and Paul, Thierry
    and Torresani, Bruno",
    title="Diffusion Maps: Using the Semigroup Property for Parameter Tuning",
    bookTitle="Theoretical Physics, Wavelets, Analysis, Genomics: An Indisciplinary Tribute to Alex Grossmann",
    year="2022",
    publisher="Springer International Publishing",
    address="Cham",
    pages="409--424",
    doi="10.1007/978-3-030-45847-8_18",
}

@article{coifman_diffusion_2026,
    title={Diffusion maps},
    author={Coifman, Ronald R and Lafon, St{\'e}phane},
    journal={Applied and computational harmonic analysis},
    volume={21},
    number={1},
    pages={5--30},
    year={2006},
    publisher={Elsevier},
    doi={10.1016/j.acha.2006.04.006}
}

@article{weng_vetra_2021,
    author = {Weng, Guangzheng and Kim, Junil and Won, Kyoung Jae},
    title = {VeTra: a tool for trajectory inference based on RNA velocity},
    journal = {Bioinformatics},
    volume = {37},
    number = {20},
    pages = {3509-3513},
    year = {2021},
    doi = {10.1093/bioinformatics/btab364},
}

@article{robrecht_dyngen_2021,
    author = {Robrecht Cannoodt and Wouter Saelens and Louise Deconinck and Yvan Saeys},
    title = {Spearheading future omics analyses using dyngen, a multi-modal simulator of single cells},
    journal = {Nature Communications},
    year = {2021},
    doi = {10.1038/s41467-021-24152-2},
}

@article{bastidasponce_comprehensive_2019,
    author = {Bastidas-Ponce, Aimée and Tritschler, Sophie and Dony, Leander and Scheibner, Katharina and Tarquis-Medina, Marta and Salinno, Ciro and Schirge, Silvia and Burtscher, Ingo and Böttcher, Anika and Theis, Fabian J. and Lickert, Heiko and Bakhti, Mostafa and Klein, Allon and Treutlein, Barbara},
    title = {Comprehensive single cell mRNA profiling reveals a detailed roadmap for pancreatic endocrinogenesis},
    journal = {Development},
    volume = {146},
    number = {12},
    pages = {dev173849},
    year = {2019},
    doi = {10.1242/dev.173849},
}

@InProceedings{maignant_tree_2025,
    author="Maignant, Elodie
    and Conrad, Tim
    and von Tycowicz, Christoph",
    editor="Nielsen, Frank
    and Barbaresco, Fr{\'e}d{\'e}ric",
    title="Tree Inference with Varifold Distances",
    booktitle="Geometric Science of Information",
    year="2025",
    publisher="Springer Nature Switzerland",
    address="Cham",
    pages="290--299",
    doi="10.1007/978-3-032-03921-7_30"
}

@article{qiu_monocle_2021,
    title={Reversed graph embedding resolves complex single-cell trajectories},
    author={Qiu, Xiaojie and Mao, Qi and Tang, Ying and Wang, Li and Chawla, Raghav and Pliner, Hannah A and Trapnell, Cole},
    journal={Nature methods},
    volume={14},
    number={10},
    pages={979–-982},
    year={2017},
    publisher={Nature Publishing Group UK London},
    doi = {10.1038/nmeth.4402},
}

@article{street_slingshot_2018,
    title={Slingshot: cell lineage and pseudotime inference for single-cell transcriptomics},
    author={Street, Kelly and Risso, Davide and Fletcher, Russell B and Das, Diya and Ngai, John and Yosef, Nir and Purdom, Elizabeth and Dudoit, Sandrine},
    journal={BMC genomics},
    volume={19},
    number={1},
    pages={477},
    year={2018},
    publisher={Springer},
    doi = {10.1186/s12864-018-4772-0},
}

@article{setty_palantir_2019,
    title={Characterization of cell fate probabilities in single-cell data with Palantir},
    author={Setty, Manu and Kiseliovas, Vaidotas and Levine, Jacob and Gayoso, Adam and Mazutis, Linas and Pe’Er, Dana},
    journal={Nature biotechnology},
    volume={37},
    number={4},
    pages={451--460},
    year={2019},
    publisher={Nature Publishing Group US New York},
    doi={10.1038/s41587-019-0068-4}
}

@article{lefort_fastme_2015,
    author = {Lefort, Vincent and Desper, Richard and Gascuel, Olivier},
    title = {FastME 2.0: A Comprehensive, Accurate, and Fast Distance-Based Phylogeny Inference Program},
    journal = {Molecular Biology and Evolution},
    volume = {32},
    number = {10},
    pages = {2798-2800},
    year = {2015},
    doi = {10.1093/molbev/msv150},
}

@article{kukurba_rna_2015,
    title={RNA sequencing and analysis},
    author={Kukurba, Kimberly R and Montgomery, Stephen B},
    journal={Cold Spring Harbor Protocols},
    volume={2015},
    number={11},
    pages={pdb--top084970},
    year={2015},
    publisher={Cold Spring Harbor Laboratory Press},
    doi={10.1101/pdb.top084970 }
}

@article{sokal_upgma_1958,
  author       = {Sokal, Robert R. and Michener, Charles D. (Charles Duncan)},
  title        = {A Statistical Method for Evaluating Systematic Relationships},
  journal      = {The University of Kansas science bulletin},
  year         = {1958},
  volume       = {38},
  number       = {22},
  doi          = {10.5281/zenodo.16435757},
}

@article{saitou_nj_1987,
    author = {Saitou, N and Nei, M},
    title = {The neighbor-joining method: a new method for reconstructing phylogenetic trees.},
    journal = {Molecular Biology and Evolution},
    volume = {4},
    number = {4},
    pages = {406-425},
    year = {1987},
    doi = {10.1093/oxfordjournals.molbev.a040454},
}

@article{ahn_eigenvalue_2013,
    author = {Ahn, Seung C. and Horenstein, Alex R.},
    title = {Eigenvalue Ratio Test for the Number of Factors},
    journal = {Econometrica},
    volume = {81},
    number = {3},
    pages = {1203-1227},
    doi = {10.3982/ECTA8968},
    year = {2013},
}

@article{watson_smooth_1964,
    author = {Geoffrey S. Watson},
    journal = {Sankhyā: The Indian Journal of Statistics, Series A (1961-2002)},
    number = {4},
    pages = {359--372},
    publisher = {Springer},
    title = {Smooth Regression Analysis},
    volume = {26},
    year = {1964},
    url = {http://www.jstor.org/stable/25049340}
}

@article{nadaraya_estimating_1964,
    author = {Nadaraya, E. A.},
    title = {On Estimating Regression},
    journal = {Theory of Probability \& Its Applications},
    volume = {9},
    number = {1},
    pages = {141-142},
    year = {1964},
    doi = {10.1137/1109020},
}

@InProceedings{networkx_2008,
  author = {Aric A. Hagberg and Daniel A. Schult and Pieter J. Swart},
  title = {Exploring Network Structure, Dynamics, and Function using NetworkX},
  booktitle = {Proceedings of the 7th Python in Science Conference},
  pages = {11 - 15},
  address = {Pasadena, CA USA},
  year = {2008},
  editor = {Ga\"el Varoquaux and Travis Vaught and Jarrod Millman},
  doi = {10.25080/TCWV9851},
}

@article{scikit-learn_2011,
  title={Scikit-learn: Machine Learning in {P}ython},
  author={Pedregosa, F. and Varoquaux, G. and Gramfort, A. and Michel, V.
          and Thirion, B. and Grisel, O. and Blondel, M. and Prettenhofer, P.
          and Weiss, R. and Dubourg, V. and Vanderplas, J. and Passos, A. and
          Cournapeau, D. and Brucher, M. and Perrot, M. and Duchesnay, E.},
  journal={Journal of Machine Learning Research},
  volume={12},
  pages={2825--2830},
  year={2011}
}

@article{yu_sequential_2021,
  title={Sequential progenitor states mark the generation of pancreatic endocrine lineages in mice and humans},
  author={Yu, Xin-Xin and Qiu, Wei-Lin and Yang, Liu and Wang, Yan-Chun and He, Mao-Yang and Wang, Dan and Zhang, Yu and Li, Lin-Chen and Zhang, Jing and Wang, Yi and others},
  journal={Cell research},
  volume={31},
  number={8},
  pages={886--903},
  year={2021},
  publisher={Springer Singapore Singapore},
  doi={10.1038/s41422-021-00486-w}
}

@article{bergen_generalizing_2020,
    title = {Generalizing RNA velocity to transient cell states through dynamical modeling},
    volume = {38},
    DOI = {10.1038/s41587-020-0591-3},
    number = {12},
    journal = {Nature Biotechnology},
    publisher = {Springer Science and Business Media LLC},
    author = {Bergen, Volker and Lange, Marius and Peidli, Stefan and Wolf, F. Alexander and Theis, Fabian J.},
    year = {2020},
    month = aug,
    pages = {1408-1414}
}

\end{document}